\documentclass[letterpaper, 10 pt]{IEEEtran}  
\IEEEoverridecommandlockouts                              
\usepackage[normalem]{ulem}
\usepackage{amsmath} 
\usepackage{amssymb}  
\usepackage{url}
\usepackage{soul,color}
\usepackage{threeparttable}
\usepackage{dcolumn}
\usepackage{mathtools}
\usepackage{comment}
\usepackage{graphicx}
\usepackage{epstopdf}
\usepackage{cite}
\usepackage{amsfonts}
\usepackage{psfig}
\usepackage{amsmath}
\usepackage{subfigure}
\usepackage{flushend}
\usepackage{times}
\usepackage{graphicx}
\usepackage{rotating}
\usepackage{enumerate}
\usepackage{epigraph}
\usepackage{float}
\usepackage{psfig}
\usepackage{epsfig}
\usepackage{graphicx}
\usepackage{amsfonts}
\usepackage{amssymb}
\usepackage{color}
\usepackage{pstricks}
\usepackage{graphics}
\usepackage{rotating}
\usepackage{lettrine}
\newcommand{\vect}[1]{\boldsymbol{#1}}
\usepackage{multicol}
\usepackage{algorithm}
\usepackage[noend]{algpseudocode}
\newtheorem{remark}{Remark}
\newtheorem{lemma}{Lemma}
\newtheorem{theorem}{Theorem}

\newtheorem{assmp}{Assumption}
\newtheorem{problem}{Problem}
\newtheorem{definition}{Definition}
\newtheorem{ex}{Example}
\usepackage{hyperref}
\title{\LARGE \bf
Finite-time Heterogeneous Cyclic Pursuit with Application to Target Interception
}

\author{Dwaipayan Mukherjee and Shashi Ranjan Kumar 
\thanks{Dwaipayan Mukherjee ({\tt\small dm@ee.iitb.ac.in}), and Shashi Ranjan Kumar ({\tt\small srk@aero.iitb.ac.in}) are Assistant Professors of Electrical Engineering, and Aerospace Engineering, respectively, at the
        Indian Institute of Technology Bombay, Mumbai- 400076, India.}
}
\makeatletter
\def\BState{\State\hskip-\ALG@thistlm}
\makeatother

\begin{document}

\maketitle
\thispagestyle{empty}
\pagestyle{empty}

\begin{abstract}
This paper presents a finite-time  heterogeneous cyclic pursuit scheme that ensures consensus among agents modelled as integrators. It is shown that for the proposed sliding mode control strategy, even when the gains corresponding to each agent are non-identical, consensus results within a finite-time provided all the gains are positive. An algorithm is presented to compute the consensus value and consensus time for a given set of gains and initial states of the agents. The set of values where consensus can occur, by varying the gains, has been derived and a second algorithm aids in determining the gains that enable consensus at any point in the aforementioned set, at a given finite-time. As an application, the finite-time consensus in line-of-sight (LOS) rates, over a cycle digraph, for a group of interceptors is shown to be effective in ensuring co-operative collision-free interception of a target, for both constant-speed and realistic models of the interceptors. Simulations vindicate the theoretical results. 
\end{abstract}


\section{Introduction}\label{intro}


In the literature pertaining to cooperative control, the consensus problem in multi-agent systems has drawn significant attention \cite{beard}. Consensus requires all or some of the states of the cooperating agents, which are dynamical systems themselves, to converge to a common value while the interactions between these agents is represented by a directed or an undirected graph. The undirected graphs, with their symmetric interactions, are more amenable to a detailed analysis, but directed graphs may be better suited to represent the interactions in a practical set-up where the agents are actual vehicles \cite{ren2008consensus}. Consensus laws and their variants have proved to be useful in many applications such as formation control, and target capture \cite{ren2006consensus, mukherjee2016target}.

One of the common directed topologies over which consensus has been studied is the directed cycle, that relates to the well-known cyclic pursuit paradigm \cite{klamkin, bruck}. While some works, such as \cite{mukherjee2016target}, have considered target interception by multiple agents in cyclic pursuit, the agents have been modelled as linear systems, mostly single or double integrators, that achieve consensus in position and velocity. Moreover, most of the consensus laws over directed cycles lead to asymptotic convergence. Ref. \cite{rao2011sliding} dealt with finite-time consensus over balanced graphs, but the consensus value could not be controlled at will since the gains for all the agents were identical and it was shown that under such circumstances the consensus value was the average of the maximum and minimum initial state values within the group. In several works, such as \cite{AS,mukherjee2016target, mukherjee2016generalized}, dealing with asymptotic versions of heterogeneous cyclic pursuit, it was shown that the ability to control the consensus value, through a tuning of edge weights, or gains, is critical to problems such as rendezvous, and target capture. However, a finite-time version of heterogeneous cyclic pursuit had not been employed to address these problems so far. Some preliminary results in this direction are provided in \cite{181}.

Further, some works \cite{mukherjee2018robustness, DZ1} have focussed on obtaining the bounds on possible perturbations to such heterogeneously chosen edge weights that still guarantee consensus. But, it may be deduced from \cite{wieland2008consensus, AS} that the consensus value, under cyclic pursuit, would be altered by any amount of perturbation on any edge. Hence, even if consensus is achieved, there is no way to ensure that the consensus value will remain insensitive to edge weight perturbations. This problem too remains relatively unexplored in the existing literature. 

A recurring feature of results related to the linear consensus problem is that their applications to any practical or realistic problem generally involve a simplification of the dynamics of individual agents. Often, the actual dynamics are either best captured by linear systems of higher order, or nonlinear systems, while the consensus results are applicable to agents modelled by integrators. A particular problem is one of cooperatively intercepting a target using multiple missiles. This is particularly significant for  intercepting many modern ships that have close-in weapon system (CIWS) \cite{tahk06} which can engage interceptors in a one-to-one scenario. By exploiting the vulnerability of CIWS to attack by multiple missiles, and its limited coverage zone, it is possible to saturate a target's engagement capability and also to intercept it along different approach angles using a team of cooperating missiles. This cooperation also enhances the visibility of target due to multiple viewing angles and is also critical to avoiding collision among the intercepting missiles.

Imposing a pre-specified impact direction in one-to-one engagements has been ensured using several strategies. One such way, that also ensures robustness against uncertainties, is the use of sliding mode control (SMC) for guidance design. For instance, \cite{tal} proposed an SMC-based guidance strategy to intercept targets in head-on, tail-chase, and head-pursuit engagements, while \cite{kumar2012sliding,kumar2014} proposed guidance laws that ensure alignment of a missile with a predefined angle, within a finite-time. Though convergence of missile heading to a desired direction of approach was ensured, an explicit evaluation of the exact time of convergence was intractable. Ref. \cite{rao} proposed an SMC based guidance which enhanced the capture zone by using two switching surfaces. However, none of these guidance laws were designed within a cooperative framework. Although \cite{vitaly} presented cooperative guidance laws that enforce certain relative geometry between missiles and target at interception, the laws are based on linearized engagement dynamics. Thus, design of co-operative interception strategies, that consider nonlinear engagement dynamics, present a challenging class of problems. 

The main contributions of the present paper can now be summarized, in the light of the above discussions. First, a finite-time version of heterogeneous cyclic pursuit is presented. The conditions for the stability of the same are analytically obtained. Next, for a given choice of gains that ensures consensus in the aforementioned cyclic pursuit paradigm, an algorithm is proposed to evaluate the value of consensus. Third, given a desired value of consensus, that belongs to the feasible set, and an a-priori selected finite-time of convergence, another algorithm is presented that provides a suitable set of gains for the agents to meet the requirements. Fourth, it is shown that a \emph{small enough} perturbation to most of the gains (other than two) does not alter the value of consensus and the finite convergence time. Fifth, the proposed finite-time cyclic pursuit set-up is applied to obtain consensus in line-of-sight (LOS) rates of multiple missiles which not only ensures successful interception of the target, but also provides a method for collision avoidance among the cooperating missiles. Finally, the proposed consensus scheme is applied to the nonlinear engagement scenario while considering both constant speed interceptors, and also interceptor models that account for variations in aerodynamic parameters leading to realistic engagement scenarios.        


This paper is organized in the following manner. Section \ref{prelims} provides some background on the existing results related to cyclic pursuit and a brief outline of sliding mode control. In Section \ref{problem} formal statements of the main problems, addressed in this paper, are given. This is followed by Section \ref{results}, where a detailed analysis of finite-time consensus over a cycle digraph is presented along with its application to the problem addressed in this paper. In Section \ref{application}, the finite-time consensus strategy over cycle digraph is applied to ensure cooperative interception of a target by multiple missiles. Section \ref{sims} establishes the validity of the presented results through simulations. Finally, Section \ref{conc} concludes the paper by outlining directions for future research.

\section{Mathematical Foundations}\label{prelims}

This section provides some background on the cyclic pursuit paradigm. Subsequently, some preliminary concepts on sliding mode control are summarized. These will aid in laying the foundations for the analysis of finite-time consensus over a directed cycle that is presented later. 

\subsection{Consensus via Cyclic Pursuit}
For a system of $n$ agents in cyclic pursuit, every agent, indexed $i$, receives information about some or all the states of an agent indexed $i+1$ (modulo $n$) and designs its control based on this information. In Fig. \ref{fig:CP} the pursuit graph is shown. Note that the flow of information is in a direction opposite to that indicated by the arrows on the edges of the pursuit graph. If the dynamics of the agents are given by single integrators, the kinematic equation of agent $i$ can be described by
\begin{align}
\dot{x}_i&=u_i,~~u_i=f_i(x_{i+1}-x_i),\label{CP_gen}
\end{align}
where $x_i$ is the state of agent $i$ (modulo $n$), $u_i$ is the control input, and $f_i(.):\mathbb{R}^n\to\mathbb{R}$ is a real valued function. When $f_i(x_{i+1}-x_i)=k_i (x_{i+1}-x_i)$, linear consensus protocol \cite{AS} results.  Under such a paradigm, the point of convergence is $X_f={\sum_{i}\left(\frac{x_i(0)}{k_i}\right)}/{\sum_{i}\frac{1}{k_i}}$,which belongs to the interior of the convex hull of the initial agent states, $\{x_i(0)\}_{i=1,\ldots,n}$, when $k_i>0~\forall~i$. Ref. \cite{AS} further showed that at most one of the gains, $k_i$, can be negative, subject to a lower bound. This expands the set of points where the agents may converge.   
\begin{figure}
\centering\includegraphics[scale=.6]{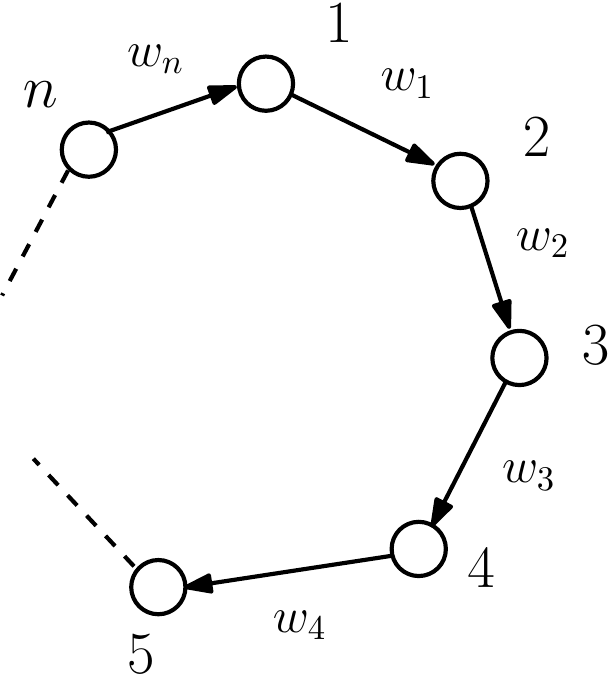}
\caption{Heterogeneous cyclic pursuit with $n$ agents.}
\label{fig:CP}
\end{figure}
\subsection{Sliding Mode Control}

Sliding mode control (SMC), which was developed for control affine dynamic systems, provides a simple design procedure and its implementation leads to such closed loop dynamics as are insensitive to parameter variations and disturbances. It involves the design of functions of the system states, and control inputs that are discontinuous in their arguments. Consider the system
\begin{equation}\label{sysdyn}
    \dot{\vect{x}} = F(\vect{x})+B(\vect{x})\vect{u},\ \vect{x}\in \mathbb{R}^n,\ \vect{u}\in \mathbb{R}^m, \ m\leq n,
\end{equation}
for which the control, $\vect{u}$, has to be designed to ensure $\vect{x}\to\vect{0}$, where $\vect{x}$ denotes the states of the system. The controller design is tantamount to choosing
\begin{enumerate}
	\item switching surfaces, $\vect{\sigma}(\vect{x})\in \mathbb{R}^m$, that are typically linear combinations of the states, and lead to $\vect{x}\to\vect{0}$ if $\vect{\sigma}(\vect{x})=\vect{0}$, 
		\item control $\vect{u}$ that acts discontinuously, component-wise:
		\begin{equation}
			{u}_i=\begin{cases}
				       u_i^+ & \textrm{if } \sigma_i(\vect{x})>0, \\
				       u_i^- & \textrm{if } \sigma_i(\vect{x})<0,
				  \end{cases}
		\end{equation}
\end{enumerate}
where $i\in\{1,\cdots, m\}$. A frequently used discontinuous function used in SMC is $\text{sign}(.):\mathbb{R}\to\{-1,0,+1\}$ (signum function), which equals $-1$, when its argument is negative, $+1$ for a positive argument and $\text{sign}(0)=0$. With an appropriate choice of control, the system states can be made to align towards $\vect{\sigma}(\vect{x})$, reach the switching surface within a finite-time, and subsequently continue to remain on this switching surface. The resultant closed loop dynamics, given by $\vect{\sigma}(\vect{x})=0$, are called the sliding modes, and are of order $m$. The main features of SMC, that distinguish it from other nonlinear control techniques, are that sliding mode is achieved within a finite-time, and that its dynamics are insensitive to system parameter perturbations, and external disturbances
\cite{VadimUtkin2009}. 
Further details about SMC are available in \cite{hungsmc,utkinsurvey}. It may be remarked that for the consensus problem, in particular, it is not straightforward to prove that SMC leads to the desired effect, i.e. consensus \cite{rao2011sliding}.

\section{Problem Formulation}\label{problem}
This paper addresses the problem of ensuring a finite-time consensus over a cycle digraph for agents modelled as single integrators. It is desirable that the designed control law, besides meeting this fundamental objective, ought to have certain other features. These are summarized in this section through formal problem statements. 
\begin{problem}\label{p2} 
For a system of $n$ agents in cyclic pursuit, whose dynamics are given by \eqref{CP_gen}, design a control law, $u_i$, that ensures consensus in the agents' states at any desired finite-time, $t_f\in(0,\infty)$, and at a desired value of consensus, $X_f$, within some \emph{feasible} set. 
\end{problem} 
\begin{problem}\label{p3}
For the designed consensus law that solves Problem \ref{p2}, determine the robustness of the consensus law against variations in control parameters. Also, determine the sensitivity of the point of convergence, $X_f$, and the time of convergence, $t_f$, to such variations.	
\end{problem}

As indicated in Section \ref{intro}, the efficacy of the designed finite-time consensus law, that solves Problem \ref{p2}, is tested under a cooperative multi-missile engagement scenario. Note that the terms `missile' and `interceptor' are used interchangeably throughout this paper. The set-up for this engagement is now described.
Planar engagement scenarios with $n$ interceptors, having equal speed $V_M$, are shown in Fig.~\ref{fig:genenggeo}. The $i^{\text{th}}$ missile's flight path angle and LOS angle are denoted by $\gamma_{M_i}$ and $\theta_{i}$, respectively, while the distance of the target from the $i^{\rm th}$ missile and the lateral acceleration of the $i^{\rm th}$ missile are $r_i$, and $a_{M_i}$, respectively. The target is at the origin of the reference frame.
\begin{figure}
\centering\includegraphics[scale=.6]{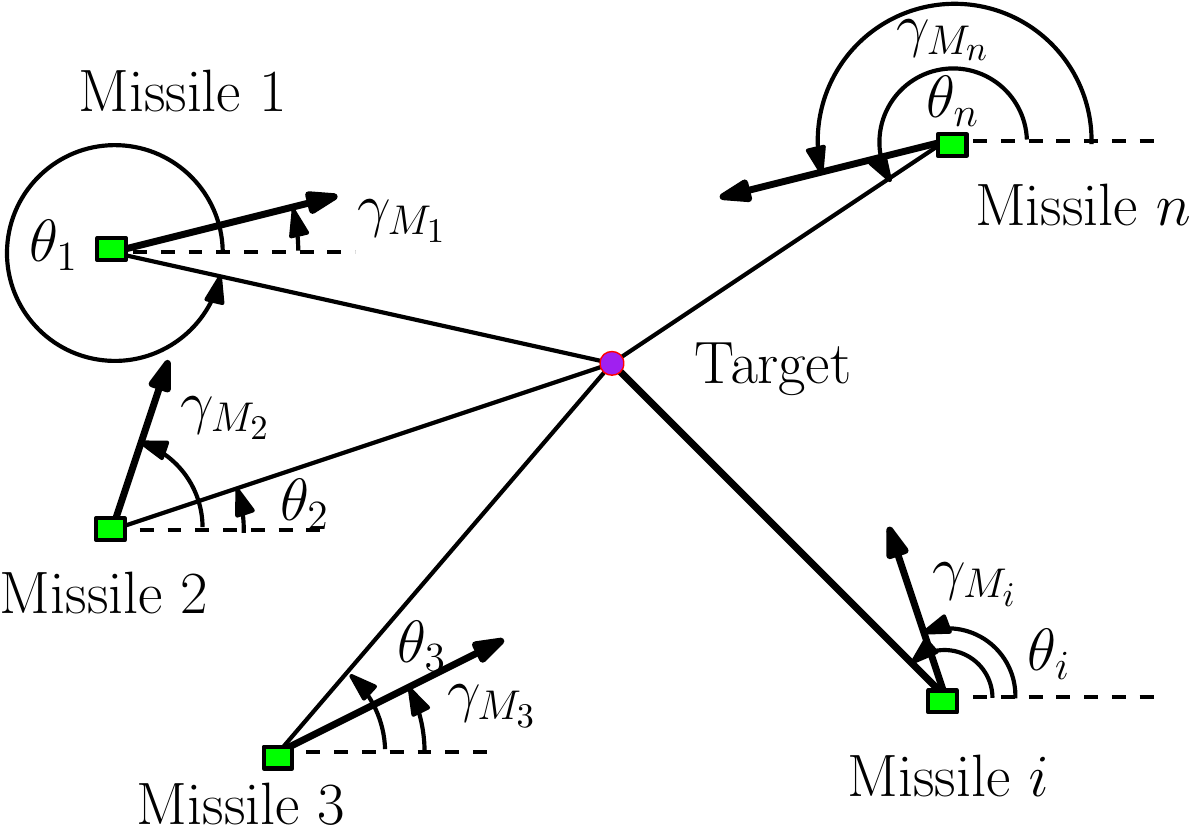}
\caption{Planar multi-missile scenario for a stationary target.}
\label{fig:genenggeo}
\end{figure}

The kinematic equation that governs the engagement between the $i^{\rm th}$ interceptor and the target is given by
\begin{subequations}\label{eq:engdyn}
\begin{align}
\dot{r}_i =&~ -V_M \cos \theta_{M_i}=~V_{r_i}, \label{eq:rdot}\\
r_i\dot{\theta}_i =&~ -V_M \sin\theta_{M_i}=~ V_{\theta_i}, \label{thetadot}\\
\dot\gamma_{M_i} =&~\frac{a_{M_i}}{V_M}\label{eq:gamadot},
\end{align}\end{subequations}
where $\theta_{M_i} = \gamma_{M_i}-\theta_i$, for $i=1,2, \ldots n$, and $V_{r_i},~V_{\theta_i}$ are the components of velocity that are parallel and perpendicular, respectively, to the corresponding LOS. The angle $\theta_{M_i}$ is the heading angle error for the target and $i^{\rm th}$ interceptor pair.

\begin{definition}(\cite{kumar2012sliding})\label{def:impactangle} Impact angle is defined as the angle, with respect to some specific reference, at which the interceptor intercepts a stationary target. In general, impact angle for $i^{\rm th}$ interceptor can be defined as the value of $\gamma_{M_{i}}$ at interception, that is, $ \theta_{\rm {imp}_i}= \gamma_{M_{iF}}$, where the subscript $F$ denotes the final value at the point of interception.
\end{definition}
\begin{remark}\label{thetadotadv}
Impact angle is the same as LOS angle, $\theta_{iF}$, at the point of interception, if $\dot\theta_{iF}=0$ and miss-distance, $r_{iF}\in(0,r_{\rm lethal})$ , where $r_{\rm lethal}$ is the lethal radius for the missile. This is so because $\dot{\theta_i}=0$ implies $\theta_{Mi}=0$, from \eqref{thetadot}, even if $r_i> 0$ and hence, $\gamma_{M_i}-\theta_i=0$.
\end{remark}
Now, some assumptions are made about the multi-interceptor engagement problem. However, these are non-restrictive.
\begin{assmp}\label{assmp1}
Guidance laws are derived under the assumption that $|\theta_{M_i}|\ne {\pi}/{2}$.
\end{assmp}

\begin{remark}
Assumption~\ref{assmp1} is not very restrictive because it has been proved in \cite{kumar2012sliding} that even if $|\theta_{M_i}|= {\pi}/{2}$ at some instant, the system will not remain there and thus such an event may occur at isolated points in time. The significance of this assumption will be clearer later when it will be shown that the control input, $a_{M_i}$, is multiplied by the term $\cos\theta_{M_i}$.
\end{remark}
\begin{assmp}\label{assmp2}
Guidance is designed within a nonlinear framework to circumvent the difficulties and limiting assumptions that arise due to linearization. 
\end{assmp}
\begin{remark}\label{relax}
Assumption \ref{assmp2} admits a broader guidance framework since it implies that the guidance designed under this nonlinear setting will also be effective in a linear framework. Moreover, the linearized dynamics are restrictive in their applicability, especially for large initial heading errors at the onset of the terminal phase.
\end{remark}
\begin{assmp}\label{assmp3}
The $n$ interceptors communicate over a directed cycle graph, as in Fig. \ref{fig:CP}, and their indexing (mod $n$) is so chosen as to ensure that $[\theta_{i+1}(0)-\theta_i(0)] (\text{mod}~ 2\pi)<[\theta_{i+2}(0)-\theta_i(0)](\text{mod}~2\pi)~\forall~i$.
\end{assmp}
\begin{remark}\label{remtheta}
Assumption \ref{remtheta} implies that the agents are indexed in ascending order starting with any one agent, labelled 1, and moving in an anti-clockwise direction with the target at the centre.
\end{remark}
\begin{assmp}\label{assmp4}
There exists at least one pair $\{k,l\}\subset\{1, 2, \ldots, n\}$ such that $\text{sign}(\dot{\theta_k}(0))=-\text{sign}(\dot{\theta_l}(0))$ where $\text{sign}(.)$ is the signum function.
\end{assmp}
\begin{remark}\label{remthetadot}
Assumption \ref{assmp4} essentially implies that the initial LOS rates of all the agents are not of the same sign. This is not too restrictive either because the sign of $\dot{\theta_i}(0)$ can be controlled by the heading angle at the onset of the terminal phase, as may be observed from \eqref{thetadot}.
\end{remark}
The third problem addressed in this paper may now be formally stated.
\begin{problem}\label{p1}
Subject to Assumptions~\ref{assmp1}, \ref{assmp2}, \ref{assmp3}, and \ref{assmp4}, design a cooperative guidance strategy to ensure the interception of a stationary target by multiple interceptors with different impact angles or approach angles, which have a relative spacing among them, while avoiding collision among the interceptors.
\end{problem}

\section{Finite-time Consensus}\label{results}
This section presents a detailed analysis of a finite-time heterogeneous cyclic pursuit scheme that leads to the solutions of Problems \ref{p2} and \ref{p3}.
In this paper, for finite-time consensus via cyclic pursuit, the function $f_i(.)$, as described in \eqref{CP_gen}, is chosen as
\begin{align}\label{HetCP_fin}
f_i(x)=w_i\text{sign}(x),
\end{align}
where $w_i>0$ and $\text{sign}(.)$ represents the signum function. The function $w_i\text{sign(.)}$ can be considered analogous to an edge weight of a cycle digraph in the linear consensus protocol \cite{mukherjee2018robustness}. It will be shown that this choice of control law results in a finite-time version of the asymptotic heterogeneous cyclic pursuit \cite{AS,mukherjee2016synchronous,mukherjee2018robustness} discussed earlier. It may be noted that a similar finite-time consensus law was analyzed for undirected graphs with identical gains (edge weights) or for balanced directed cycle graphs, that is, with $w_i=M>0~\forall~i$ in \cite{rao2011sliding}. In contrast, heterogeneous gains are admitted in this paper for a class of digraphs: directed cycles. 
The system is thus described by:
\begin{align}\label{mainCP}
\dot{x}_i(t)=-w_i\text{sign}(\sigma_i(t)),
\end{align}
where $\sigma_i(t)=x_i(t)-x_{i+1}(t)$.
\begin{lemma}\label{signlemma}
At any instant of time, $t$, unless $\sigma_i(t)=0~\forall ~i$, there exist some $k,l\in\{1,2,\ldots,n\}$ such that $\text{sign}(\sigma_k(t))=-\text{sign}(\sigma_l(t))$.
\end{lemma}
\begin{proof}Observe that $\sum_{i=1}^{n}\sigma_i(t)=0~\forall t$. Now, if $\text{sign}(\sigma_k(t))=\text{sign}(\sigma_l(t))~\forall~k,l$, then $\sigma_i(t)$ must be zero for all $i$. This completes the proof.
\end{proof}
\begin{lemma}\label{state_evol}
If $\sigma_i(t)$ does not change sign over the interval $(t_0,t_1)$ and $x_i(t)\neq x_{i+1}(t)$ for any $t\in (0,t_0)$, then $x_i(t)$ for $t\in(t_0,t_1)$ is given by:
\begin{align}\label{state}
x_i(t)=
\begin{cases}
x_{i}(t_0)-w_i t~\text{for}~\sigma_i(t_0)>0\\
x_{i}(t_0)+w_i t~\text{for}~\sigma_i(t_0)<0.
\end{cases}
\end{align}
\end{lemma}
\begin{proof}
Proof follows from integrating both sides of \eqref{mainCP} using the fact that $\text{sign}(\sigma_i(t))$ is a constant ($+1$ or $-1$) over the interval $(t_0,t_1)$.
\end{proof}
\begin{remark}\label{rem1}
It follows from Lemma \ref{state_evol} that at any instant of time, $t$, the evolution of the state $x_i(t)$ is given by a segment of a straight line, the sign of whose slope depends on the sign of $\sigma_i(t)$ and the magnitude of this slope is $w_i$. Without loss of generality, suppose the state $x_i(t)$ evolves along a straight line segment, say $\mathcal{L}_{i}(t)$, with a slope $+w_i$. Evidently, as long as $\mathcal{L}_i(t)$ does not intersect with $\mathcal{L}_{i+1}(t)$ at some $t=t_i$, the slope of $\mathcal{L}_i$ remains unchanged.
\end{remark}

\begin{lemma}\label{intersect2}
For $w_i<w_{i+1}$, intersection of $\mathcal{L}_i(t)$ and $\mathcal{L}_{i+1}(t)$ at some $t=t_i>0$ may occur only if $\text{sign}(\sigma_{i+1}(t_i^-))=-\text{sign}(\sigma_{i}(t_i^-))$ and consequently, $\text{sign}(\sigma_{i}(t_i^+))=\text{sign}(\sigma_{i+1}(t_i^+))=\text{sign}(\sigma_{i+1}(t_i^-))$.
\end{lemma}
\begin{proof}
Consider $V_i=\dfrac{1}{2}\sigma_i(t)^2\geq0~\forall t$. Differentiating both sides with respect to time, it follows that $$\dot{V_i}= |\sigma_i(t)|\text{sign}(\sigma_i(t))[w_{i+1}\text{sign}(\sigma_{i+1}(t))-w_{i}\text{sign}(\sigma_{i}(t))].$$ Evidently, $\dot{V}_i(t)>0$ when $\text{sign}(\sigma_{i+1}(t))=\text{sign}(\sigma_{i}(t))$ and $w_i<w_{i+1}$. Thus, $|x_{i}(t)-x_{i+1}(t)|$ is an increasing function of time, thereby precluding an intersection of $\mathcal{L}_i(t)$ and $\mathcal{L}_{i+1}(t)$. However, if $\text{sign}(\sigma_{i+1}(t))=-\text{sign}(\sigma_{i}(t))$, then $\dot{V}_i(t)<0$, and $|x_{i}(t)-x_{i+1}(t)|$ is a decreasing function of time, even if $w_i<w_{i+1}$. This can result in an intersection of $\mathcal{L}_i(t)$ and $\mathcal{L}_{i+1}(t)$ at some $t=t_i$, where $x_i(t_i)=x_{i+1}(t_i)$. But, just after intersection, at $t=t_i^+$, while the slope of the segment $\mathcal{L}_{i+1}$ remains unchanged, unless it also intersects with $\mathcal{L}_{i+2}$ at $t_i$, there is a sign reversal of $\sigma_i(t)$ at $t=t_i$. This can be observed by comparing the signs of  $\sigma_i(t_i^-)$ and $\sigma_i(t_i^+)$. Without loss of generality, suppose $\text{sign}(\sigma_{i+1}(t_i^-))=-\text{sign}(\sigma_{i}(t_i^-))=\text{sign}(\sigma_{i+1}(t_i^+))=1$.
For some infinitesimal $\delta>0$, such that $t_i^+=t_i+\delta$, $x_i(t_i^+)=x_i(t_i)+w_i\delta$ and $x_{i+1}(t_i^+)=x_{i+1}(t_i)-w_{i+1}\delta$ (from Lemma \ref{state_evol}). Since $x_i(t_i)=x_{i+1}(t_i)$, it follows that $\text{sign}(\sigma_{i}(t_i^+))=\text{sign}(\delta (w_i+w_{i+1}))=-\text{sign}(\sigma_i(t_i^-))=1$. Thus, sign reversal occurs at $t_i$. Similar arguments hold if $\text{sign}(\sigma_{i+1}(t_i^-))=-\text{sign}(\sigma_{i}(t_i^-))=\text{sign}(\sigma_{i+1}(t_i^+))=-1$.
\end{proof}
\begin{figure}
\centering\includegraphics[scale=.75]{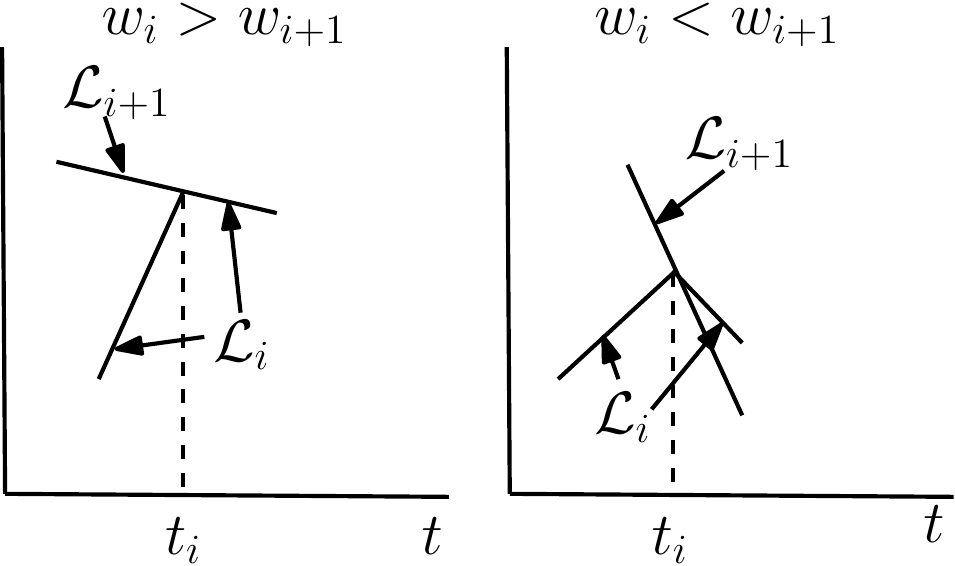}
\caption{Intersection of $\mathcal{L}_i$ and $\mathcal{L}_{i+1}$.}
\label{Intsect}
\end{figure}
\begin{lemma}\label{intersect3}
For $w_i=w_{i+1}$, intersection of $\mathcal{L}_i(t)$ and $\mathcal{L}_{i+1}(t)$ at some $t=t_i>0$ may occur and after intersection, $x_i(t)=x_{i+1}(t)$ for all $t>t_i$ if $\text{sign}(\sigma_{i+1}(t_i^-))=-\text{sign}(\sigma_{i}(t_i^-))$, while for $\text{sign}(\sigma_{i+1}(t_i^-))=\text{sign}(\sigma_{i}(t_i^-))$ $\mathcal{L}_i(t)$ and $\mathcal{L}_{i+1}(t)$ do not intersect but run parallel to each other.
\end{lemma}
\begin{proof}
For $w_i=w_{i+1}$ it readily follows that so long as $\text{sign}(\sigma_{i+1}(t_i^-))=\text{sign}(\sigma_{i}(t_i^-))$, $\mathcal{L}_i(t)$ and $\mathcal{L}_{i+1}(t)$ are segments of two parallel lines. However, for $\text{sign}(\sigma_{i+1}(t_i^-))=-\text{sign}(\sigma_{i}(t_i^-))$, it follows that $\dot{V}_i(t)=0$ at $t=t_i$, where $\mathcal{L}_i$ may possibly intersect $\mathcal{L}_{i+1}$. Using the same arguments as in the proof of Lemma \ref{intersect2}, it may be observed that there is a sign reversal of $\sigma_i(t)$ at $t_i$ and thereafter for all $t>t_i$, $\dot{V}_i(t)=0$ holds.
\end{proof}
\begin{lemma}\label{intersect1}
For $w_i>w_{i+1}$, $\mathcal{L}_i(t)$ may intersect $\mathcal{L}_{i+1}(t)$ at some $t=t_i>0$, and after intersection $x_i(t)=x_{i+1}(t)$ for all $t>t_i$.
\end{lemma}
\begin{proof}
Consider $V_i$ as defined in the proof of Lemma \ref{intersect2}. Note that here it is assumed that $\mathcal{L}_{i+1}$ has a slope, $\pm w_{i+1}$ at the time $t_i$. If, however, $\mathcal{L}_{i+1}$ has already intersected and merged with $\mathcal{L}_{i+2}$, then its slope will be $\pm w_{i+2}$. In any case, this merger is only possible when $w_{i+2}\leq w_{i+1}$ (as is evident from Lemmas \ref{intersect2} and \ref{intersect3}), and subsequently, $w_{i+2}<w_{i}$. In that case, replacing $w_{i+1}$ by $w_{i+2}$ will not alter this proof. One may thus write 
\begin{align}\label{lemmaV}\dot{V_i}= -w_i|\sigma_i(t)|\left[1-\dfrac{w_{i+1}}{w_i}\text{sign}(\sigma_{i+1}(t))\text{sign}(\sigma_{i}(t))\right]\end{align}
which, along with $w_i>w_{i+1}$, leads to $\dot{V_i}<0~\forall t$, irrespective of the value of $\text{sign}(\sigma_{i+1}(t))\text{sign}(\sigma_{i}(t))$. Further, $\dot{V_i}<0~\forall t$ ensures that $\mathcal{L}_i(t)$ intersects $\mathcal{L}_{i+1}(t)$ at $t=t_i$. Thus, $\sigma_i(t)=0~\forall t>t_i$. 
\end{proof}
\begin{remark}\label{2stickornot2stick}
As may be observed from Fig. \ref{Intsect}, Lemmas \ref{intersect2} and \ref{intersect1} indicate that if $w_i>w_{i+1}$, then $\mathcal{L}_i(t)$ and $\mathcal{L}_{i+1}(t)$ intersect at $t_i>0$ and thereafter $\mathcal{L}_i(t)$ merges with $\mathcal{L}_{i+1}(t)$. On the other hand, for $w_i<w_{i+1}$, even if $\mathcal{L}_i(t)$ and $\mathcal{L}_{i+1}(t)$ intersect at $t_i>0$, thereafter $\mathcal{L}_i(t)$ does not merge with $\mathcal{L}_{i+1}(t)$. Instead, the slope of $\mathcal{L}_i(t)$ changes sign at $t=t_i$, while retaining its magnitude, $w_i$.
\end{remark}
\begin{theorem}\label{convergence}
The cyclic pursuit system given by \eqref{CP_gen}, and \eqref{HetCP_fin} achieves consensus in its states within a finite-time if $w_i>0~\forall~i$.
\end{theorem}
\begin{proof}
At an instant $t$, there always exist $X_m(t)=\min_{i}x_i(t)$, and $X_M(t)=\max_{i}x_i(t)$. Consider the sets $S_m(t)=\{i:x_{i}(t)=X_m(t)\}$, and $S_M(t)=\{i:x_{i}(t)=X_M(t)\}$. Unless $X_m(t)=X_M(t)$, in which case consensus has been achieved and $\dot{x}_i(t)=0~\forall~i$, it follows that $S_m\cap S_M=\emptyset$. Consider the function $V(t)=X_M(t)-X_m(t)$. For any agent, $i\in S_M(t)$, it follows that $\sigma_i(t)\geq0$ and thus, from Lemma \ref{state_evol}, $\dot{x}_i\leq 0~\forall i\in S_M(t)$. Similarly, for all $i\in S_m(t)$, one may conclude that $\dot{x}_i\geq 0$. Assume that $X_m(t)\neq X_M(t)$, and $\dot{x}_i= 0~\forall~i\in S_m\cup S_M$. This is not possible since for some $i\in S_m$, there exists $i+1\in \{1,2,\dots,n\}\setminus S_m$, so that $\dot{x}_i>0$. Consequently, for every agent $i\in S_m$, $\dot{x}_i>0$ since otherwise the corresponding $\mathcal{L}_i$ will not stick to $\mathcal{L}_{i+1}$, thereby violating Lemma \ref{intersect1}. Analogously, there exists some $i\in S_M$ such that $i+1\in \{1,2,\dots,n\}\setminus S_M$, so that $\dot{x}_i<0$. This, in turn, leads to $\dot{x}_i<0~\forall~i\in S_M$. Thus, except at some isolated points (sets of measure zero), where two segments $\mathcal{L}_i$ and $\mathcal{L}_{i+1}$ intersect, the function $V(t)$ is differentiable everywhere else (differentiable almost everywhere). At these isolated points, Clark's generalized gradient \cite{shevitz1994lyapunov} may be applied to conclude that $\dot{V}(t)<0$ at these sets of measure zero. Further, wherever $V(t)$ is differentiable, there too $\dot{V}(t)<0$. This implies that the cyclic pursuit system eventually achieves consensus. It remains to be shown that this consensus is achieved in finite-time. Consider the two gains whose values are the lowest, say $w_k$ and $w_q$ without loss of generality. It follows that since the evolutions of the states, $x_i(t)$ are along segments of straight lines whose slopes are equal to the gains, $\{w_i\}_{i=1,2,\cdots,n}$ and the values of these slopes may switch from one gain to another, for an agent, at time instants $t_i$, therefore $X_m(t)$ and $X_M(t)$ must also evolve likewise. Thus, $\dot{V}(t)<-(w_k+w_q)<0~\forall~t$ and $\dot{V}(t)$ is piecewise constant. Consequently, the system attains consensus at or before the time given by $t_c={V(0)}/(w_k+w_q)$.
\end{proof}

\begin{remark}\label{reachable}
Theorem \ref{convergence} and its proof establish the existence of some finite-time, $t_f$, and some consensus value, $X_f$. However, in order to constructively obtain the point of convergence and the time of convergence, Algorithm \ref{algo} is effective. The convergence of the algorithm follows from Theorem \ref{convergence}.
\end{remark}
Some terms need to be defined before Algorithm \ref{algo} is stated. All the indices of the agents are taken modulo $n$. 

\begin{definition}\label{def:set1}If the segments $\mathcal{L}_i(t)$ and $\mathcal{L}_{i+1}(t)$ intersect at time $t_i$ for the first time, then $X_i=x_i(t_i)=x_{i+1}(t_i)$, and $\mathcal{P}_i$ is used to denote the pair $(t_i,X_i)$.
\end{definition}
Further, at any instant of time, $t$, define the set $\mathcal{J}(t)$ as
\begin{align}\label{setJ}\mathcal{J}(t)=\{i:\mathcal{L}_i(t)=\mathcal{L}_{i+1}(t)\},
\end{align} 
and its cardinality is given by $|\mathcal{J}(t)|=r$. Suppose the elements of $\mathcal{J}(t)$ are given by $\mathcal{J}(t)=\{j_1,j_2,\ldots, j_r\}$. Algorithm \ref{algo} may now be  presented.
\begin{algorithm}
\caption{Determination of $t_f$ and $X_f$}\label{algo}
\begin{algorithmic}[1]
\State $i=1, t=0$.
\BState \emph{begin loop}
\State Obtain $ \mathcal{P}_{i}$ as in Definition \ref{def:set1}; $i=i+1$.
\State If $i < n$, goto 3.
\BState \emph{end loop}
\State Obtain $\mathcal{P}_n$ with $\mathcal{L}_n(t)$ and $\mathcal{L}_1(t)$.
\State Construct the set $\mathcal{P}=\{\mathcal{P}_i\}_{i=1,2, \ldots, n}$ .
\State Obtain $ T=\min_{i:t_i-t>0} t_i-t$.
\State Construct the set $\mathcal{J}(t+T)$ as in \eqref{setJ}.
\State $k=1$.
\State If $w_{j_k}\geq w_{j_k+1}$, assign $\mathcal{L}_{j_k}(\tau)=\mathcal{L}_{j_k+1}(\tau), \tau>t+T $;  
else $\mathcal{L}_{j_k}(\tau)=x_{j_k}(t+T)-w_{j_k}\text{sign}(\sigma_{j_k}[(t+T)^-])\text{sign}(\sigma_{j_k}[(t+T)^+])\tau, \tau>t+T $.
\State Update $\mathcal{P}_{j_k-1}$, ($j_k-1 ~\text{mod}~ n$) and subsequently $\mathcal{P}$.
\State $k=k+1$.
\State If $k<|\mathcal{J}(t+T)|+1$, goto 11.
\State $t=t+T$.
\State $X_M=\max_i x_i(t)$, $X_m=\min_i x_i(t)$.
\State If $X_M-X_m>0$, goto 8 
\State $X_f=X_M$, $t_f=t$.
\end{algorithmic}
\end{algorithm}
\begin{remark}\label{algconv}
Note that in step 11 of Algorithm \ref{algo}, a segment $\mathcal{L}_{j_k}$ merges with $\mathcal{L}_{j_k+1}$ if $w_{j_k}\geq w_{j_k+1}$, thereby reducing the number of such segments by one during the next iteration that commences at step 8. However, if $w_{j_k}< w_{j_k+1}$, the number of segments do not reduce. Hence, if every intersection of a segment $\mathcal{L}_{j_k}$ and $\mathcal{L}_{j_k+1}$ is such that $w_{j_k}< w_{j_k+1}$, Algorithm \ref{algo} will fail to converge. But, owing to the cyclic nature of the graph, one cannot have all the gains in such an order that $w_i<w_{i+1}~\forall~i$. Hence, after $n$ iterations at most, at some $t=t_{j_k}$, one of the segments, $\mathcal{L}_{j_k}$, must merge with $\mathcal{L}_{j_k+1}$. Similarly, with this updated $\mathcal{L}_{j_k}$, there can be at most $n-1$ iterations after which at least one of the segments merges with its leading segment. By the principle of mathematical induction, it follows that after at most $\dfrac{n(n-1)}{2}$ iterations, Algorithm \ref{algo} must converge. 
\end{remark}
\begin{remark}\label{lastconv}
Note that the last intersection of some $\mathcal{L}_{j_k}$ and $\mathcal{L}_{j_k+1}$ at $t=t_f$ essentially means that all the agents approach each other in two groups (either merged with $\mathcal{L}_{j_k}$ or with $\mathcal{L}_{j_k+1}$). Suppose $w_{j_k}<w_{j_k+1}$. It would be erroneous to conclude, using Lemma \ref{intersect2}, that there will be a sign reversal of the slope of $\mathcal{L}_{j_k}$ and further divergence. This is because in the proof of Lemma \ref{intersect2} it was considered that $\mathcal{L}_{i+1}(t)$ continues unaltered after the intersection with $\mathcal{L}_i(t)$. However, at $t_f$ all the switching functions are zero, implying $\dot{x}_i(t)=0~\forall~i$.
\end{remark} 
\begin{remark}\label{reach_set}
From the proof of Theorem \ref{convergence}, owing to the strictly monotonic nature of $X_M(t)$ and $X_m(t)$, it follows that the point of convergence, $X_f\in (X_m(0),X_M(0))$. However, explicit closed form expressions for the point of convergence, as in \cite{AS} for asymptotic cyclic pursuit, cannot be obtained for an arbitrary choice of positive gains, $w_i$. Instead, Algorithm \ref{algo} provides a way of computing both $X_f$ and the convergence time, $t_f$, for a given set of positive gains.
\end{remark}

Conversely, suppose it is required to choose a set of gains $\{w_i\}_{i=1,2,\ldots,n}$ to ensure that the convergence occurs at a point $X_f\in (\min_i x_i(0),\max_i x_i(0))$ and at time $t_f$. The choice of gains is non-unique and in this paper one gain-selection method to achieve consensus at the desired point and at a given time is provided in the form of an algorithm (Algorithm \ref{algo2}). Before Algorithm \ref{algo2} is presented, some relevant sets need to be defined. Suppose $X_m(0)=\min_i x_i(0)$ and $X_M(0)=\max_i x_i(0)$, which is consistent with the notation in proof of Theorem \ref{convergence}. It will be assumed that $x_p(0)\neq x_q(0)~\forall p,q\in\{1,2,.\ldots,n\}$, although, with some minor refinements, Algorithm \ref{algo2} works even when this assumption does not hold. Thus the sets $S_M(0)$ and $S_m(0)$, as defined in the proof of Theorem \ref{convergence}, are singletons, say $S_M(0)=\{i\}$ and $S_m(0)=\{j\}$. Now consider a partition of the vertex set, $\mathcal{V}$, such that $\mathcal{V}=\mathcal{V}_{max}\cup\mathcal{V}_{min}$, and $\mathcal{V}_{max}\cap\mathcal{V}_{min}=\emptyset$, where $\mathcal{V}_{max}=\{j+1,j+2,\ldots,i-1,i\}$, and $\mathcal{V}_{min}=\{i+1,i+2,\ldots,j-1,j\}$. This partitioning is illustrated in Fig. \ref{partition}, where the dashed line partitions the vertex set into $\mathcal{V}_{max}$ (the top half) and $\mathcal{V}_{min}$ (the bottom half). Suppose the cardinalities of $\mathcal{V}_{max}$ and $\mathcal{V}_{min}$ are $r$ and $n-r$, respectively. 

\begin{figure}[!ht]
\centering\includegraphics[scale=.5]{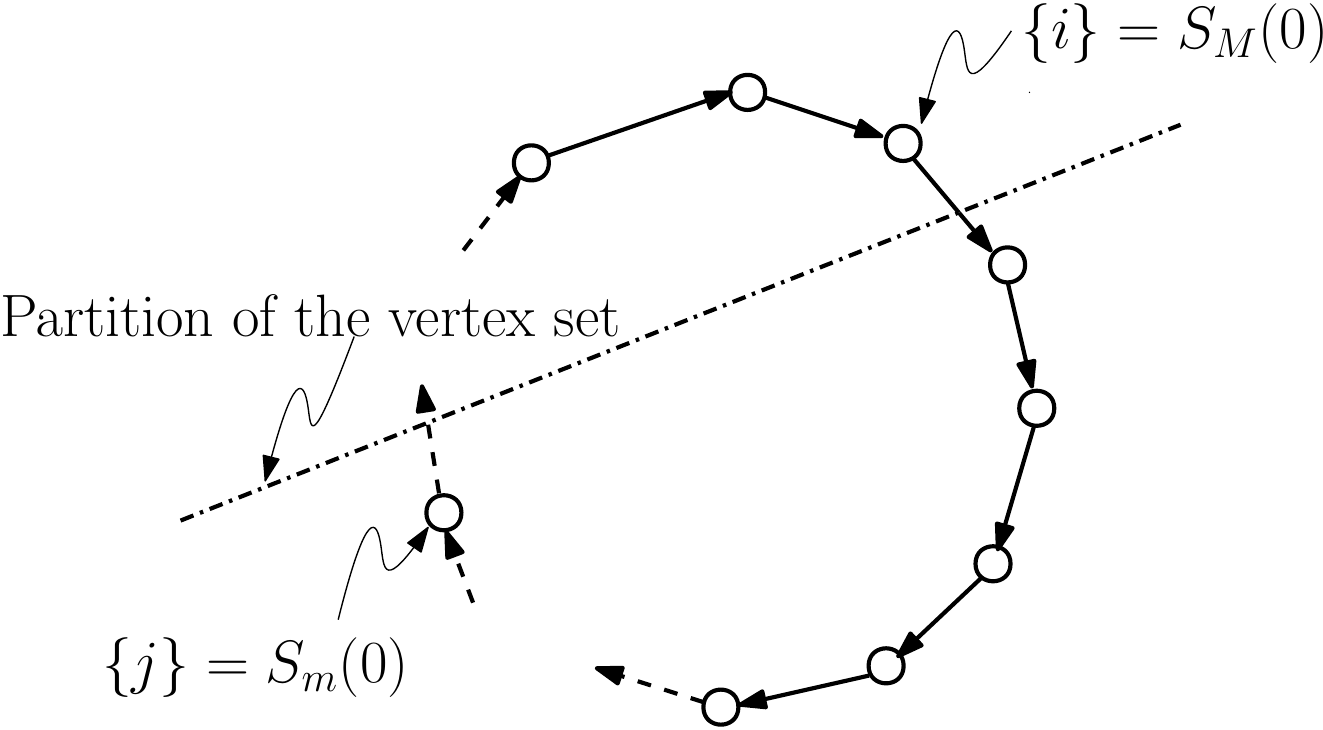}
\caption{A partition of the vertex set used in Algorithm \ref{algo2}.}
\label{partition}
\end{figure}

\begin{algorithm}
\caption{Choice of gains for a given $t_f$ and $X_f$}\label{algo2}
\begin{algorithmic}[1]
\State For $S_M(0)=\{i\}$, and $S_m(0)=\{j\}$, choose $w_i=\dfrac{x_i(0)-X_f}{t_f}$ and $w_j=\dfrac{X_f-x_j(0)}{t_f}$.
\State $k=1$, $l=1$ $t_i=t_j=t_f$.
\BState \emph{begin loop}
\State Choose $w_{j-k}$ such that $w_{j-k}>w_{j-k+1}$ and $\dfrac{|\sigma_{j-k}(0)|}{w_{j-k}-\dfrac{w_{j-k+1}\text{sign}(\sigma_{j-k+1}(0))}{\text{sign}(\sigma_{j-k}(0))}}<t_{j-k+1}$.
\State $t_{j-k}=\dfrac{|\sigma_{j-k}(0)|}{w_{j-k}-\dfrac{w_{j-k+1}\text{sign}(\sigma_{j-k+1}(0))}{\text{sign}(\sigma_{j-k}(0))}}$.
\State $k=k+1$
\State If $k < n-r$, goto 3.
\BState \emph{end loop}
\BState \emph{begin loop}
\State Choose $w_{i-l}$ such that $w_{i-l}>w_{i-l+1}$ and $\dfrac{|\sigma_{i-l}(0)|}{w_{i-l}-\dfrac{w_{i-l+1}\text{sign}(\sigma_{i-l+1}(0))}{\text{sign}(\sigma_{i-l}(0))}}<t_{i-l+1}$.
\State $t_{i-l}=\dfrac{|\sigma_{i-l}(0)|}{w_{i-l}-\dfrac{w_{i-l+1}\text{sign}(\sigma_{i-l+1}(0))}{\text{sign}(\sigma_{i-l}(0))}}$.
\State $l=l+1$
\State If $l < r$, goto 9.
\BState \emph{end loop}
\end{algorithmic}
\end{algorithm}

\begin{lemma}\label{alglemma}
For the finite-time heterogeneous cyclic pursuit system \eqref{CP_gen}, and \eqref{HetCP_fin} whose gains, $\{w_i\}_{i=1,2,\ldots,n}$, are chosen using Algorithm \ref{algo2}, the following hold:
\begin{align}
\label{monotone}
S_M(t_1)\subseteq S_M(t_2),~\text{and}~S_m(t_1)\subseteq S_m(t_2)~\forall~t_1\leq t_2,
\end{align}
where $S_M(t)$ and $S_m(t)$ are defined in the proof of Theorem \ref{convergence}. Also, $\sigma_i(t),~\forall~i, t<t_f$, does not reverse its sign. 
\end{lemma}
\begin{proof}
Note that $S_M(0)=\{i\}$ and $S_m(0)=\{j\}$ are non-empty sets. Subsequently, for any $t>0$, Lemmas \ref{intersect2}, \ref{intersect3}, and \ref{intersect1}, combined with the gain selection method in Algorithm \ref{algo2} ensure that $w_i<w_{i-1}<\ldots<w_{j+1}$. This implies that the segments $\mathcal{L}_{i-1},~\mathcal{L}_{i-2},\ldots,\mathcal{L}_{j+1}$ each merge with $\mathcal{L}_{i}$ before $t_f$ at instants given by $t_{i-1}>t_{i-2}>\dots>t_{j+1}$, and $i-k\in S_M(t)~\forall t>t_{i-k}$ where $k=1,2,\ldots, i-j+1 (\text{mod}~n)$. Similar arguments hold for $S_m(t)$. This proves \eqref{monotone}.

Note that since Algorithm \ref{algo2} ensures that $w_i<w_{i-1}<\ldots<w_{j+1}$ and $w_j<w_{j-1}<\ldots<w_{i+1}$, with indices chosen modulo $n$, all intersections between $\mathcal{L}_{p}(t)$ and $\mathcal{L}_{p+1}(t)$, except the last one between $i$ and $j$ correspond to the scenario in Fig. \ref{Intsect} for $w_i>w_{i+1}$ and hence, by Lemma \ref{intersect1}, $\sigma_i(t)~\forall~i$ never reverses its sign.    
\end{proof}
\begin{remark}\label{robustness}\emph{(Insensitivity to gain perturbations)}
Observe that if gains are chosen using Algorithm \ref{algo2}, there always exists $\delta>0$, such that the final point of convergence, $X_f$, is not altered when the gain $w_k\in(\bar{w}_{k}-\delta,\bar{w}_{k}+\delta)$, corresponding to the agent $k$, where $k\notin S_M(0)\cup S_m(0)$, and $\bar{w}_{k}$ corresponds to the nominal value of the gain. 
\end{remark}

%

Using homogeneous gains, consensus will always be achieved at the average of the maximum and minimum of the agents' initial states. But in the statement of Problem \ref{p2}, it is stated that the system is required to achieve consensus at a value within some \emph{feasible} set and so the homogeneous finite-time cyclic pursuit presented in \cite{rao2011sliding} is not applicable. Although using heterogeneous gains, it has been shown through Algorithms \ref{algo} and \ref{algo2} that the set of points where consensus can occur is not a singleton, it is still not possible to reach consensus at any arbitrary point on the state space. Hence, a characterization of this \emph{feasible} set within which the desired point of consensus must lie, needs to be characterized. The following result deals with such \emph{reachable sets}  
\begin{lemma}\label{reachableset}\emph{(Reachable set)} Consider the $d$-dimensional Euclidean space, $\mathbb{R}^d$, over which the dynamics of the agents are described. Suppose the dynamics along the different principle axes are decoupled. The point of consensus belongs to the set $\mathcal{R}=(X_{1m}(0),X_{1M}(0))\times (X_{2m}(0),X_{2M}(0))\times \cdots\times (X_{dm}(0),X_{dM}(0))$, where $X_{rm}(0)=\min_{i}x_{ri}(0)$,and $X_{rM}(0)=\max_{i}x_{ri}(0)$, $r\in\{1,2,\cdots,d\}$. 
\end{lemma}
\begin{proof} The proof follows from the observation that the $r^{\rm th}$ coordinate ($r=1,2,\cdots,d$) of the point of convergence belongs to $(X_{rm}(0),X_{rM}(0))$ and that the dynamics along the principle axes are decoupled, implying that the gains along these directions can be independently designed.  
\end{proof}

The sufficiency of $w_i>0~\forall i$ in order to guarantee consensus (according to Theorem \ref{convergence}) indicates that like the asymptotic heterogeneous cyclic pursuit in \cite{AS}, consensus may be achievable even with a negative gain. The set of reachable points may also be expanded in this case by using negative gain for any chosen agent. Theorem \ref{convergence} also implies that even if the gains are perturbed from their original values, as long as these gains are all positive, consensus will not be disrupted. However, it is interesting to investigate whether consensus is achievable when one of the gains has a negative value, and if so, then what the lower limit of such a negative gain is. This question is one of robustness, similar to the study of edge weight perturbations in \cite{mukherjee2018robustness, DZ1}, and the analysis of reachable sets, if consensus is achievable with a negative gain, is beyond the scope of this paper.  The following theorem answers this robustness question.
\begin{theorem}\label{neggain}
For the system described by \eqref{CP_gen}, and \eqref{HetCP_fin}, suppose a gain $w_k<0$ while the remaining gains are all positive. Let $W_{\rm min}=\min_{i,i\neq k} w_i$. The system will achieve consensus iff $w_k>-W_{\rm min}$.
\end{theorem}
\begin{proof}
Suppose at least one of the gains, $w_k<0$. Consequently, $\mathcal{L}_i$ can only intersect $\mathcal{L}_{i+1}$ at some $t=t_i$ iff $\text{sign}({\sigma_i(t_i^-))}=-\text{sign}(\sigma_{i+1}(t_i^-))$ and $|w_k|<\text{slope of $\mathcal{L}_{i+1}$ at $t_i$}$. This is the only way $\dot{V}_i(t)<0$ can be ensured. Now, applying similar reasoning as in the proof of Lemma \ref{intersect2}, it may be concluded that at $t_i$, $\sigma_i(t)$ reverses sign and consequently, $\dot{V}_i(t_i^+)>0$. Thus, $V_i(t)$ starts increasing again. Due to the cycle topology, it follows that eventually the segments, $\mathcal{L}_i$ corresponding to all the remaining agents in $\mathcal{V}\setminus \{k\}$ will try to intersect $\mathcal{L}_k$ either individually or in clusters. However, $\mathcal{L}_k$ will still have a slope of $\pm w_k$. Thus, any segment (or multiple merged segments) that can intersect $\mathcal{L}_k$ must have a slope greater than $|w_k|$. After every merger of an $\mathcal{L}_i$ with $\mathcal{L}_{i+1}$ (for $w_i>w_{i+1}$), the slope of the merged line is always $w_{i+1}$, the lower of the two gains. Thus, among all these segments that approach $\mathcal{L}_k$, the one with the least slope is $W_{\rm min}$. Therefore, $w_k>-W_{\rm min}$ must hold for a intersection that may lead to consensus. This proves the necessity.

Suppose the agent $k$ with $w_k>-W_{\rm min}$ is the only one with a negative gain. Consensus may fail to occur only if the segment $\mathcal{L}_i$ corresponding to some agent $i$ fails to merge with $\mathcal{L}_k$. Now, consider all agents other than $k$. These agents will form a chain and whether each of these agents, say $i$, merges with their leader $i+1$ depends on whether $w_i>w_{i+1}$ or not. The leader of this chain is $k-1$ and its tail is $k+1$. Even if these remaining agents do not merge with their corresponding leaders, they will move towards their leader since their gains are all positive. It therefore follows that eventually all these remaining agents will approach agent $k$ with different gains. Since after a successful merger, both agents move with the lower of the two gains, $w_{i+1}$, it follows that some agents will approach agent $k$ with a gain equal to $W_{\rm min}$. Since, $|w_k|<W_{\rm min}$, all these agents will be able to catch up with agent $k$. However, they may or may not merge with $k$ depending on whether agent $k-1$ is part of their cluster. By induction, it follows that subsequently all the agents will merge with $k$. Even agent $k+1$, which may approach $k$ alone or in a cluster will merge with $k$, once $k+2$ has merged with $k$, by following its leader. After the mergers are all complete, $\dot{x}_i(t)=0~\forall~ i$, resulting in consensus. This proves the sufficiency.  
\end{proof}
\begin{ex}
Consider five agents whose initial states are given by $10^3\times$[1~2~ -5~ 4~ 6]. Suppose the nominal gains are $\vect{w}$=[20~40~30~15~ 25]. Consider a perturbation to gain $w_5$. According to Theorem~\ref{neggain}, $W_{\rm min}$ for this case is $15$. The top and bottom portions of Fig.~\ref{figneggain} show the evolutions of switching surfaces when $w_5$ is perturbed to $-14.9$ and $-15.1$, respectively.  This illustrates that $-W_{\rm min} = -15$ is the critical value of gain $w_5$ at which consensus breaks down. 

\begin{figure}[!h]
	\centering\includegraphics[scale=.48]{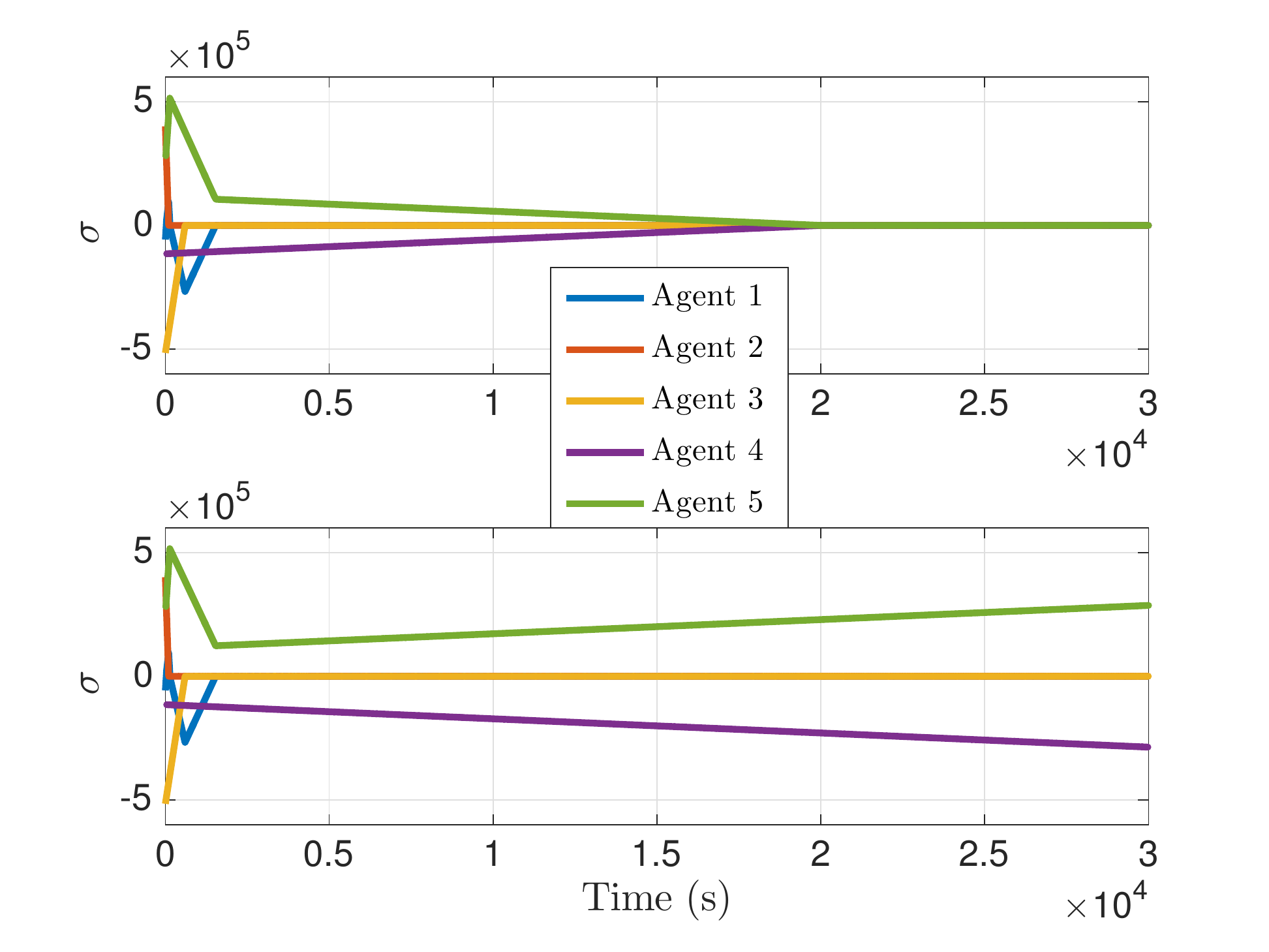}
		\caption{Convergence and divergence of switching surface around critical values of gain, $-W_{\min}$}
	\label{figneggain}
\end{figure}
\end{ex}

\section{Cooperative Guidance Scheme}\label{application}
In this Section, the results on finite-time cyclic pursuit, presented in Section \ref{results}, will first be utilized to solve Problem \ref{p1}. Next, the merits of cooperation among the intercepting missiles will be elucidated.
\subsection{Guidance Strategy}\label{guide} 
The multi-missile engagement described in Section \ref{problem} will be mathematically analyzed here in order to aid in the design of a suitable guidance strategy.
\begin{lemma}\label{lew_LOSr}
The dynamics of the LOS angles for missile-target pairs are of relative degree two with respect to the lateral accelerations of the missiles, and is given by:
\begin{equation}\label{LOSdyn}
\begin{aligned}
\ddot \theta_i  =  - \frac{{2\dot r_i\dot \theta_i }}{r_i } - \left(\frac{{\cos \theta_{M_i } }}{r_i }\right)a_{M_i}
\end{aligned}
\end{equation}
\end{lemma}
\begin{proof}
On differentiating \eqref{thetadot} with respect to time, and substituting suitably using \eqref{eq:gamadot} and \eqref{eq:rdot}, it follows that
\begin{align}
\dot r_{i}\dot\theta_{i}+r_{i}\ddot\theta_{i} 
=&-a_{M_i}\cos\theta_{M_i}-\dot r_{i}\dot\theta_{i}.\label{eq:rlambdadmdot11}
\end{align}
Clearly, the dynamics of $\theta_{i}$ is of second degree with respect to missile lateral acceleration. The result is true for any $i\in \{1\cdots n\}$ and this completes the proof.
\end{proof}

\begin{remark}
Although the LOS angle $\theta_i$ has a relative degree of two with respect to the control input $a_{M_i}$, the LOS rate $\dot \theta_i$ has a relative degree of one with respect to the same. Note that if the LOS rates for all agents are zero then a successful interception of target is guaranteed provided closing velocity, $V_{c_i}= -V_{r_i}>0~\forall~i$.
\end{remark}
\begin{figure}[!ht]
	\centering\includegraphics[scale=.5]{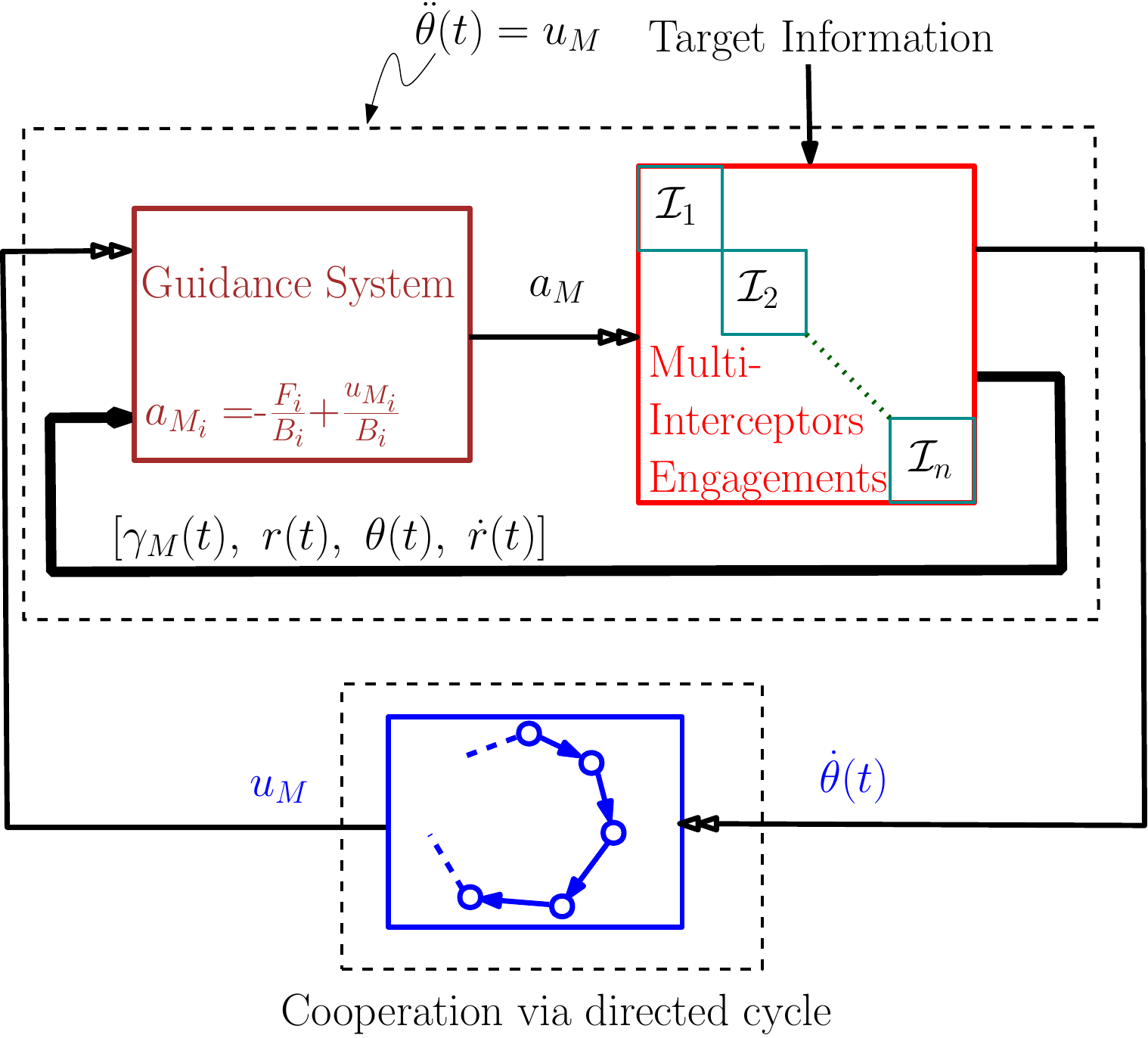}
	\caption{Schematic representation of cooperative guidance system.}
	\label{guidcycle}
\end{figure}
Note that the dynamics of $\theta_i$, for the $i^{\text{th}}$ interceptor, can be represented in a form similar to \eqref{sysdyn} where 
\begin{align}
F_i=&~ - \frac{{2\dot r_i\dot \theta_i }}{r_i},~~B_i= ~-\frac{{\cos \theta_{M_i } }}{r_i }.\label{LOSdynrep}
\end{align}
\begin{lemma}
If the lateral acceleration of $i^{\rm th}$ interceptor is chosen as
\begin{align}
a_{M_i}=&~-\dfrac{F_i-u_{M_i}}{B_i}=~  -\dfrac{{2\dot r_i\dot \theta_i }}{\cos \theta_{M_i } }-\dfrac{ r_iu_{M_i}}{\cos\theta_{M_i}},\label{amtrf}
\end{align}
then the dynamics of corresponding LOS rate, $\dot{\theta}_i$, reduces to that of a single integrator as in \eqref{CP_gen}, with the auxiliary control input $u_{M_i}$.
\end{lemma}
\begin{proof}
On substituting \eqref{amtrf} in \eqref{LOSdyn}, one obtains
\begin{align}
\nonumber \ddot \theta_i  =&  - \frac{{2\dot r_i\dot \theta_i }}{r_i } +\frac{{\cos \theta_{M_i } }}{r_i }\left[\dfrac{{2\dot r_i\dot \theta_i }}{\cos \theta_{M_i } }+\dfrac{ r_iu_{M_i}}{\cos\theta_{M_i}}\right]=u_{M_i}
\end{align}
This completes the proof.
\end{proof}

\begin{theorem}\label{guidtheo}
Consider the dynamics of LOS rates given by \eqref{LOSdyn}. 
If the control input, $u_{M_i}$, is chosen as
\begin{align}\label{mainCPmiss}
u_{M_i}=-w_i\text{sign}(\sigma_i(t)),
\end{align}
where $\sigma_i(t)=\dot\theta_i(t)-\dot\theta_{i+1}(t)$, so that the guidance command for $i^{\rm th}$ interceptor as
\begin{align}\label{guidmiss}
a_{M_i}=\dfrac{{2\dot V_{c_i}\dot \theta_i }(t)}{\cos \theta_{M_i } }+\dfrac{ r_iw_{i}}{\cos\theta_{M_i}}\text{sign}(\dot\theta_i(t)-\dot\theta_{i+1}(t)),
\end{align}
there exists a choice of $\{w_i\}_{i=1,2,\ldots,n}$ which ensures that the LOS rates, $\dot{\theta}_i~\forall i$, converge to zero within a finite-time, provided Assumption \ref{assmp4} holds.
\end{theorem}
\begin{proof}
By substituting the value of $a_{M_i}$ and $u_{M_i}$ in \eqref{LOSdyn}, one obtains dynamics similar to that of \eqref{CP_gen}. The proof of convergence of LOS rates follows from Theorem \ref{convergence}. Further, Algorithm \ref{algo2} ensures a suitable choice of $\{w_i\}_{i=1,2,\ldots,n}$ since Assumption \ref{assmp4} guarantees that $0\in (\dot{\theta}_{min}(0),\dot{\theta}_{max}(0))$.
\end{proof}
Note that in the statement of Theorem \ref{guidtheo}, Assumption \ref{assmp4} may be removed by the use of negative gains to achieve consensus at zero LOS rate even when $0\notin (\dot{\theta}_{min}(0),\dot{\theta}_{max}(0))$. But since this paper does not discuss any method to design non-positive gains to achieve consensus at a point outside $(\dot{\theta}_{min}(0),\dot{\theta}_{max}(0))$, Assumption \ref{assmp4} plays a crucial role.

A schematic representation of the overall guidance scheme for the multi-missile system is shown in Fig.~\ref{guidcycle}.

\begin{remark}
The guidance command for the $i^{\rm th}$ missile, \eqref{guidmiss}, consists of two components. One is responsible for the interception of the target while the other ensures consensus in LOS rates at a fixed finite-time. Also, the first term is reminiscent of proportional navigation (PN) class of guidance strategies. Note that the second term depends on the relative LOS rates between the leader ($i+1^{\text{th}}$ missile), and the follower ($i^{\text{th}}$ missile), in a cycle digraph, and not on their absolute values.
\end{remark}

\subsection{Merits of Cooperation}
Among several benefits of cooperation, avoidance of collision among the missiles is most relevant to Problem \ref{p1}.
\begin{lemma}\label{colavoid}
	\emph{(Collision Avoidance)}
	For gains designed using Algorithm \ref{algo2}, the guidance command \eqref{guidmiss} guarantees inter-agent collision avoidance if $\left|\int_{0}^{t_f}\sigma_i(\tau) d\tau\right|<|\theta_{i+1}(0)-\theta_i(0)|~\forall~i$.
\end{lemma} 
\begin{proof}
A sufficient condition for collision avoidance is that no two missiles ever have the same LOS angle during the engagement. This, due to Assumption \ref{assmp3}, implies that $$\left|\int_{0}^{t}\sigma_i(\tau) d\tau\right|<|\theta_{i+1}(0)-\theta_i(0)|~\forall~i~\text{and}~\forall~t\leq t_f.$$ For $t>t_f$, all the missiles have the same LOS rate, zero in the present case, and so collision is ruled out. Now, by Lemma \ref{alglemma}, Algorithm \ref{algo2} implies that $\left|\int_{0}^{t}\sigma_i(\tau) d\tau\right|~\forall~i$ is a non-decreasing function of $t\in(0,t_f)$, since the sign of $\sigma_i(\tau)$ does not reverse in this interval. Thus, one has $\left|\int_{0}^{t}\sigma_i(\tau) d\tau\right|\leq \left|\int_{0}^{t_f}\sigma_i(\tau) d\tau\right|~ \forall~ t\leq t_f~\text{and}~\forall~i$. Hence, $\left|\int_{0}^{t_f}\sigma_i(\tau) d\tau\right|<|\theta_{i+1}(0)-\theta_i(0)|~\forall ~i$ guarantees collision avoidance.
\end{proof}
\begin{remark}\label{rem:colavoid}
The consensus time, $t_f$, can be made arbitrarily small by scaling up all the gains by the some factor, while the consensus value is unchanged. By choosing suitably high values of gains, it can always be ensured that $\left|\int_{0}^{t}\sigma_i(\tau) d\tau\right|<|\theta_{i+1}(0)-\theta_i(0)|$ for all $t\leq t_f$, even if Algorithm \ref{algo2} is not used. But this comes at a price, since higher values of gains will lead to higher lateral acceleration demands. 
\end{remark}

Once consensus in LOS rates is achieved, the lateral acceleration requirement for each missile is zero. Hence, it may be tempting to reduce the consensus time, $t_f$, arbitrarily. However, this will lead to requirement of higher lateral acceleration for $t<t_f$, and may even result in its saturation. Therefore, a trade off is required. 

\begin{remark}As explained in Remark \ref{thetadotadv}, the consensus of $\dot{\theta}_i~\forall i$ at zero results in alignment of all the interceptors along the fixed final angles $\gamma_{{M_iF}}=\theta_{iF}$ for all $t>t_f$ (that is, before the engagement is over), which also corresponds to the impact angle. This is also a feature of the cooperative scheme.
\end{remark}
If multiple non-cooperating missiles intercept a target individually within a fixed finite-time, say by extending the work in \cite{kumar2012sliding,kumar2014} to multiple missiles, the exact time instants when all the missiles intercept the target is not easy to compute. This is, however, not true of the cooperative approach since after $t_f$, the missiles move along straight lines towards the target. Finally, even if the consensus value of LOS rates is chosen to be non-zero, the collision avoidance feature of Algorithm \ref{algo2} is retained. Note that in the proof of Lemma \ref{colavoid}, the consensus value of $\dot{\theta}_i$ does not influence the subsequent steps. 

%

\begin{figure}[!h]
	\centering \subfigure[Trajectories of missiles]{\includegraphics[scale=.57]{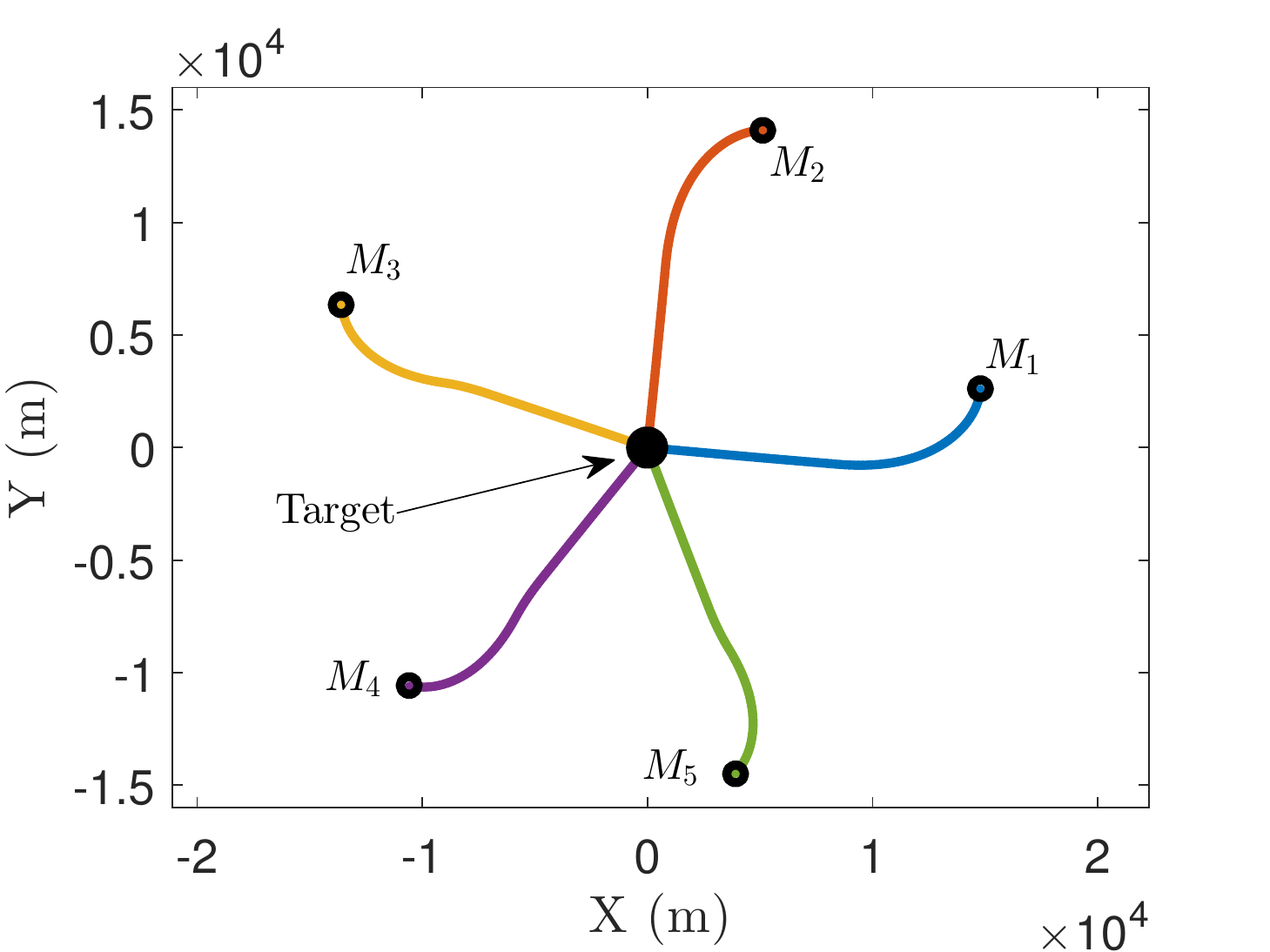}\label{trex1}}\hspace{-0.5cm}
	\subfigure[Switching surface deviations and lateral accelerations]{\includegraphics[scale=.45]{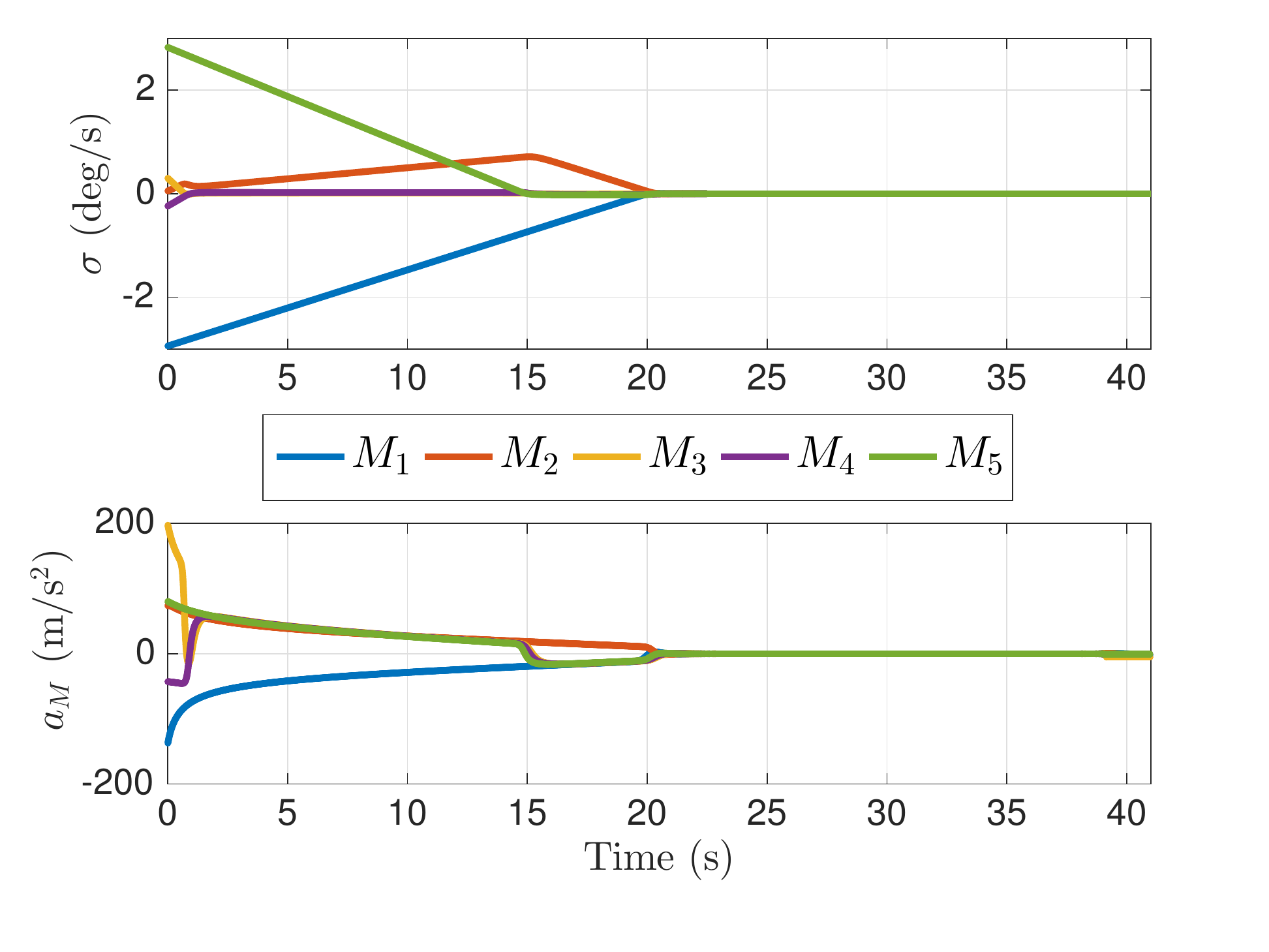}\label{samex1}}\hspace{-0.5cm}
	\subfigure[LOS angles and their rates]{\includegraphics[scale=.57]{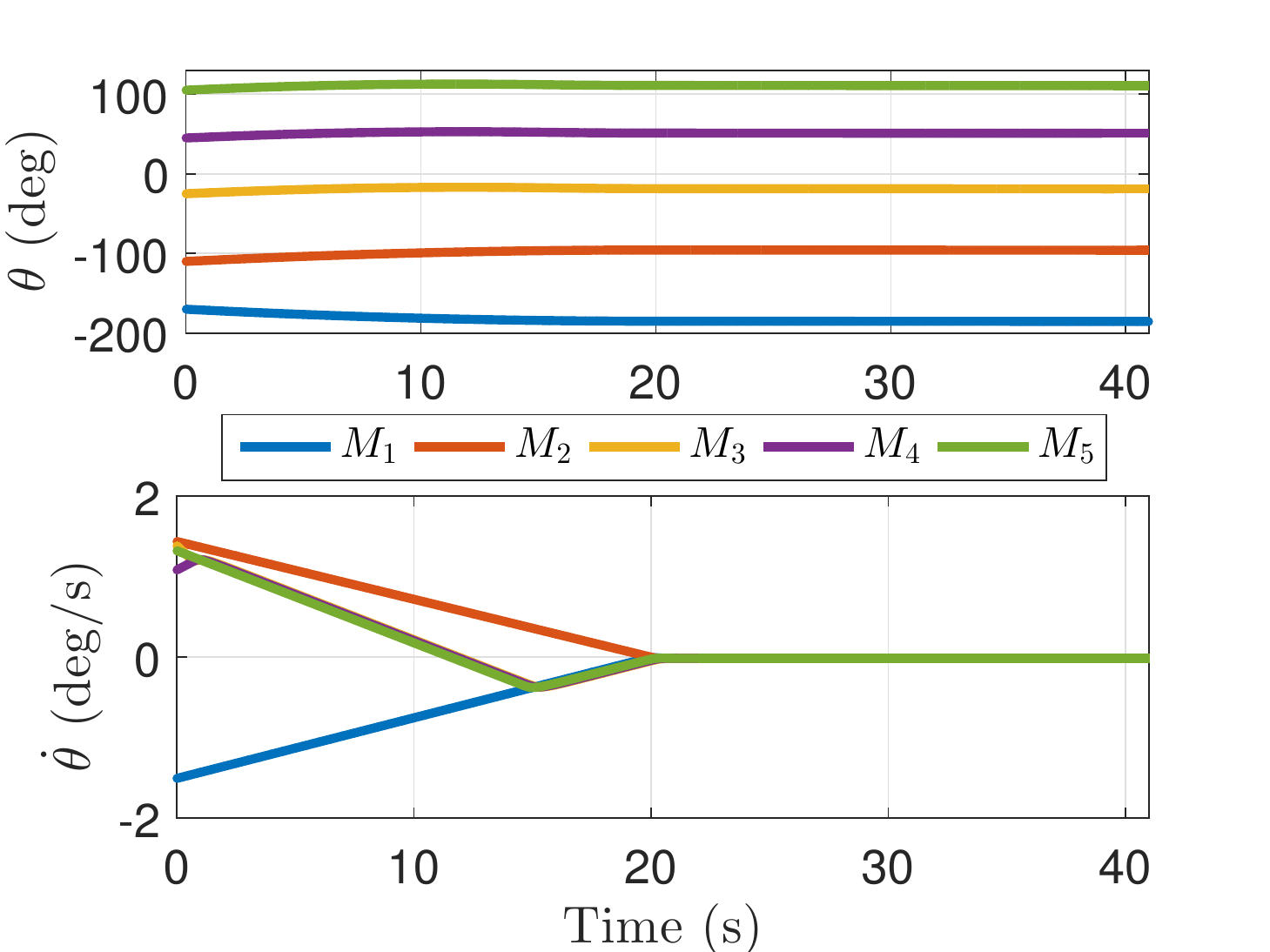}\label{thex1}}\hspace{-0.5cm}
	\caption{Cooperative guidance design using Algorithm~\ref{algo2}.}
	\label{figex1}
\end{figure}
\begin{figure}[!ht]
	\subfigure[Trajectories of missiles]{\includegraphics[scale=.57]{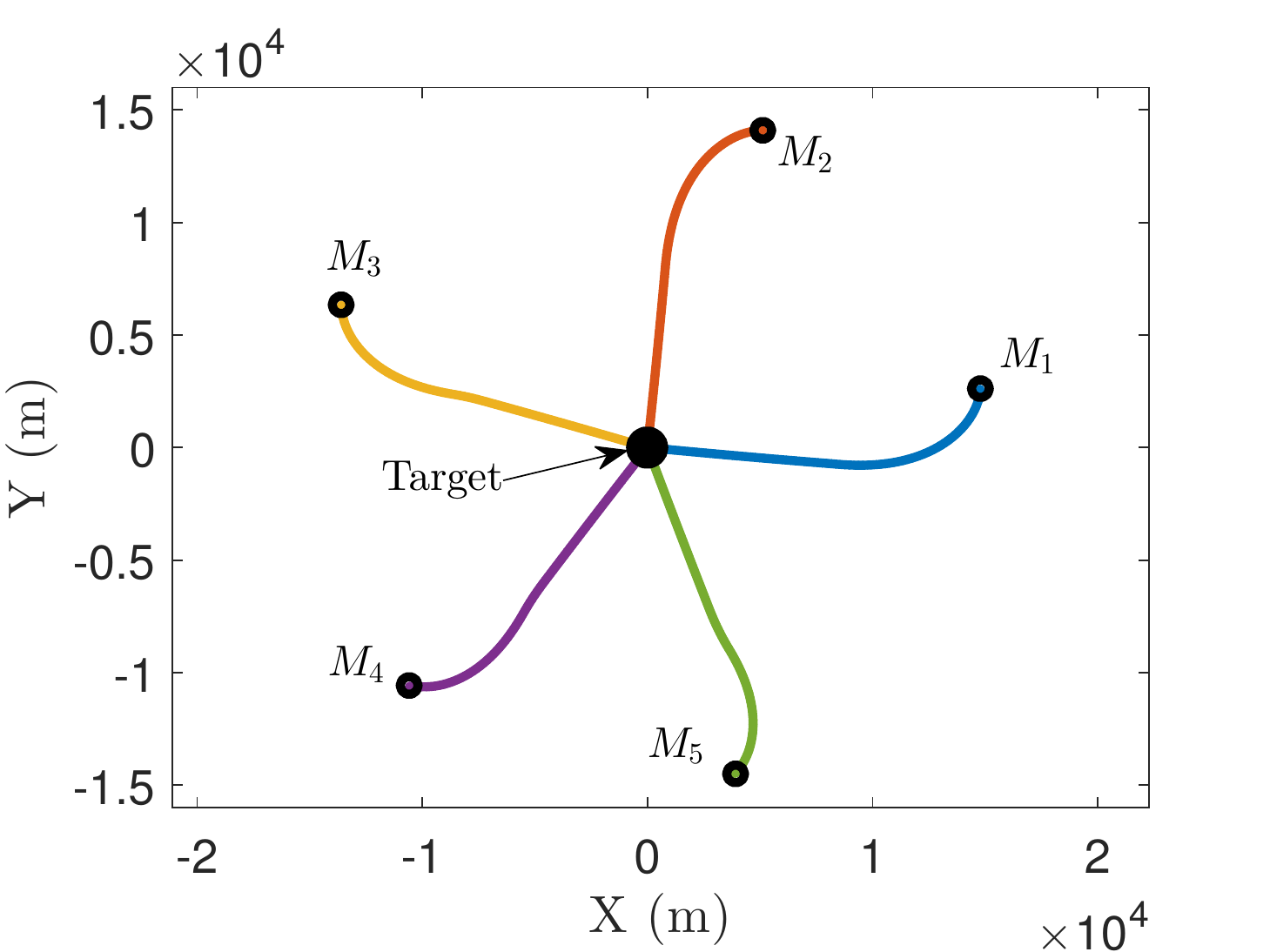}\label{trex2}}\hspace{-0.5cm}
	\subfigure[Switching surface deviations and lateral accelerations]{\includegraphics[scale=.57]{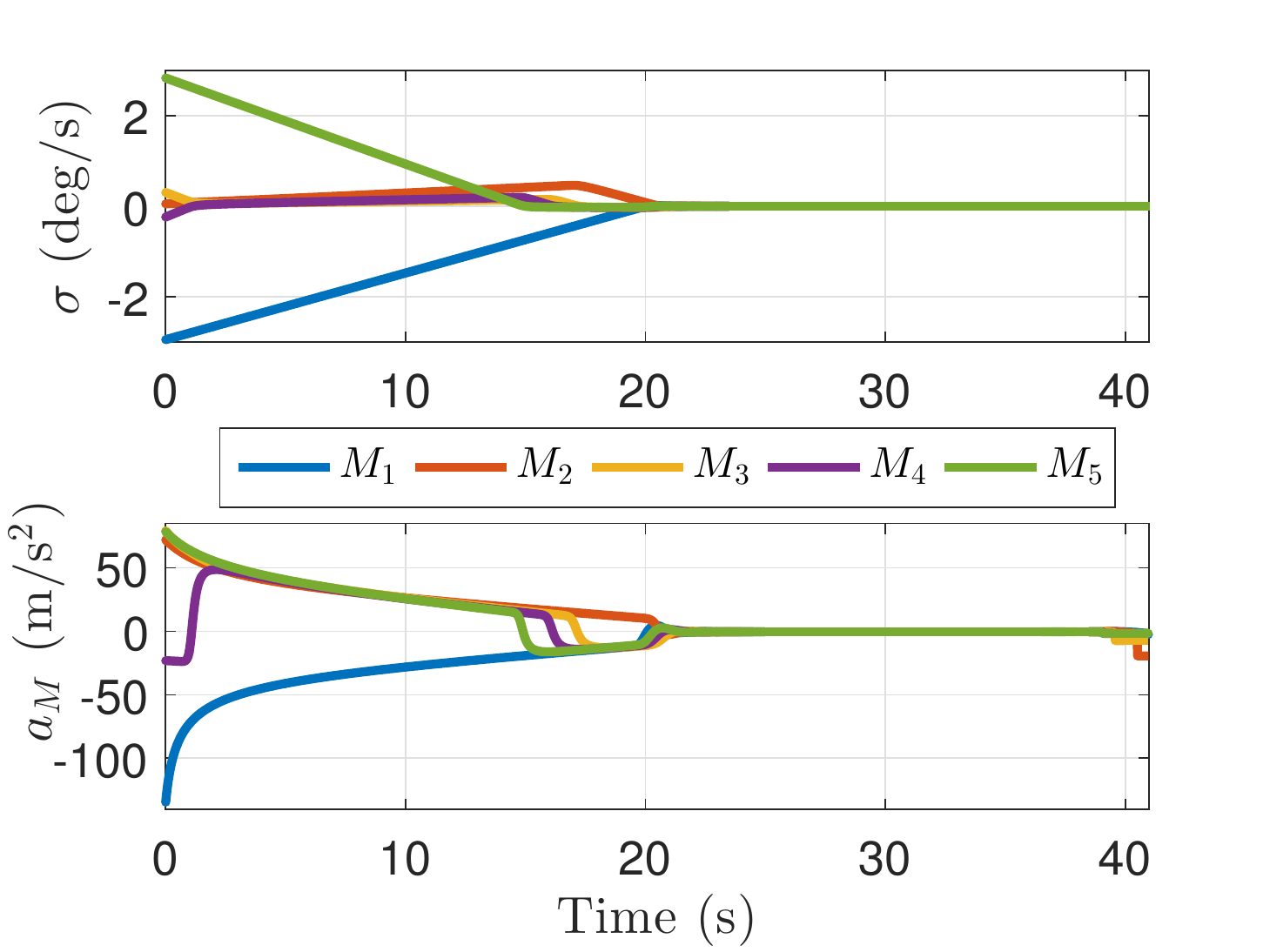}\label{samex2}}\hspace{-0.5cm}
	\subfigure[LOS angles and their rates]{\includegraphics[scale=.57]{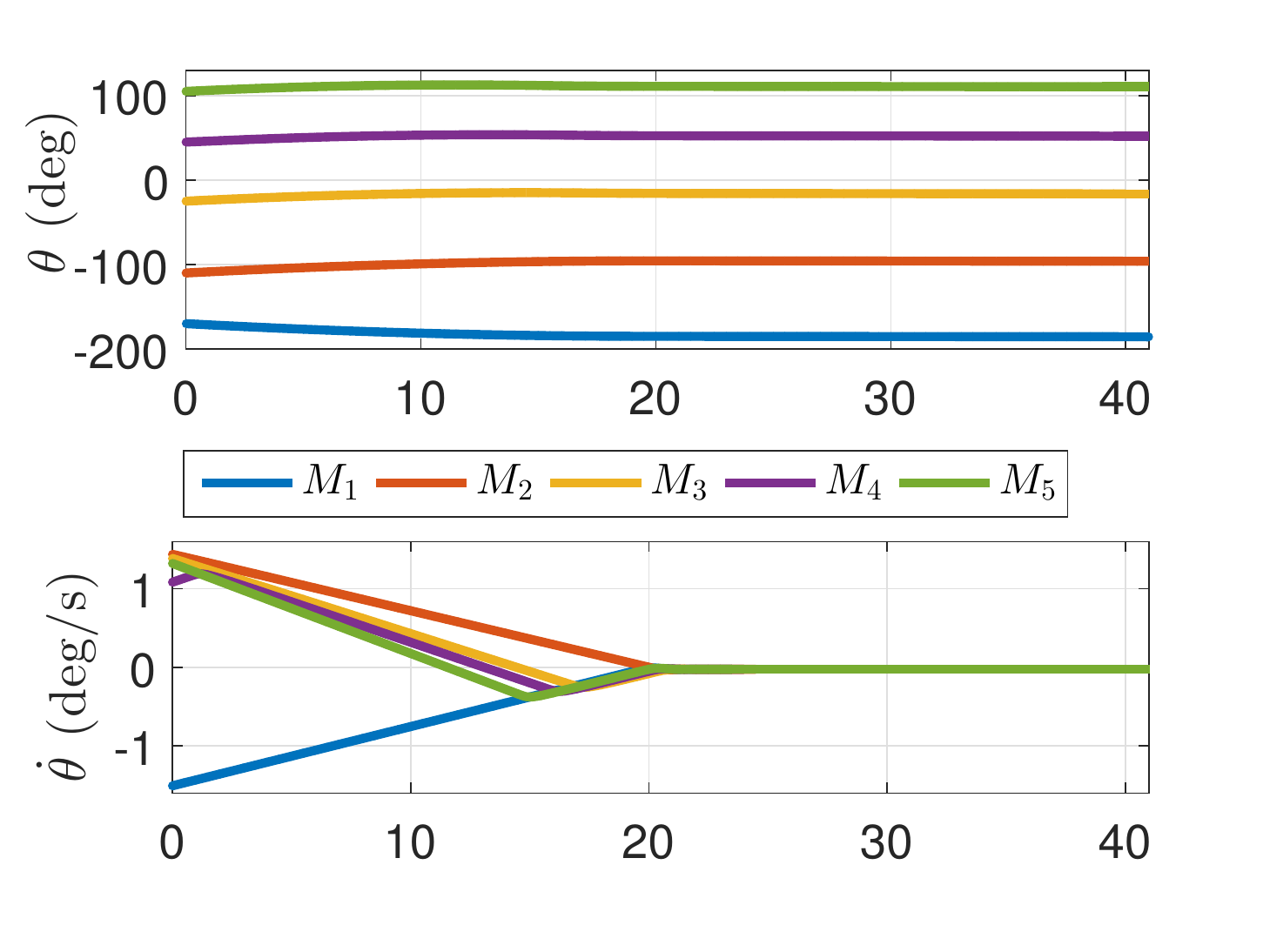}\label{thex2}}\hspace{-0.5cm}
	\caption{Cooperative guidance design without Algorithm~\ref{algo2}. }
	\label{figex2}
\end{figure}
\section{Simulations}\label{sims}
Simulations are carried out for both constant speed interceptors, and for interceptors whose speeds vary with time owing to aerodynamic variations. However, the gains of the interceptors are computed, using Algorithm \ref{algo2}, assuming constant speeds for the interceptors, in either case. Therefore, successful interception of targets for the latter case would validate the efficacy of the proposed cooperative guidance scheme.

\subsection*{Constant Speed Missile Models}
Two scenarios are considered here with five missiles, each of whose speed, $V_M$, is 400~m/s. The range of each missile in either example is 15~km and the values of initial LOS angles $\vect{\theta}(0) = [-170^\circ ~ -110^\circ~  -25^\circ~  45^\circ ~ 105^\circ]$. Since the earliest possible interception time is $r(0)/V_M = 37.5$~s, the consensus in LOS rates is desired at 20~s. Once consensus is achieved, impact direction will be same as the LOS angle which is consistent with discussions in Section~\ref{problem}. Due to Assumption~\ref{assmp4}, $0\in(\min_i\{\dot\theta_i(0)\},~\max_i\{\dot\theta_i(0)\})$. So Algorithm~\ref{algo2} guarantees a choice of gains, $\{w_i\}_{i=1\ldots n}$ that ensures consensus at zero LOS rate. In the first engagement scenario, the requisite gains are chosen by trial and error to ensure consensus at zero LOS rate, while Algorithm~\ref{algo2} is employed for the same in the second case.

\subsubsection*{Scenario~1}
The initial flight path angles of the five missiles are $[-90^\circ~-180^\circ~-90^\circ~0^\circ~45^\circ]$, which result in initial LOS rates given by [-2.6262~  2.5058~  2.4168 ~   1.8856 ~   2.3094]$\times 10^{-2}$ rad/s. With gains chosen as $\vect{w}$ = [1.315~ 1.255~ 4.900 ~2.700 ~2.000]$\times 10^{-3}$, using Algorithm~\ref{algo2}, consensus is achieved at 20 seconds. Simulation results for this case are shown in Fig. \ref{figex1}, which depicts the trajectories of the five missiles, switching surface and corresponding their lateral acceleration profiles, the evolution of their LOS angles, and their rates with time. From Fig \ref{trex1}, it follows that all the five missiles intercept the target while avoiding collision among themselves. Fig. \ref{samex1} shows that the missiles require higher lateral acceleration values before consensus is reached, but upon reaching consensus, the lateral acceleration requirement is zero. Note that if the gains are reduced, leading to a slower consensus, the maximum lateral acceleration required by each missile can be reduced. After consensus in LOS rates is reached, both the LOS rates and the LOS angles remain invariant, as illustrated in Fig. \ref{thex1}. Note that the invariance of LOS angles is a result of LOS rate consensus at zero. Further, no two missiles have the same LOS angles at any instant of time, thereby ruling out the possibility of collision.


\subsubsection*{Scenario~2}
The initial flight path angles are the same as in Scenario 1, leading to the same initial LOS rates. To illustrate the robustness of proposed algorithm with respect to variation in gains, the perturbed gains are chosen as $\vect{\bar{w}}=[1.315~ 1.255~ 1.685~ 1.800~ 2.000]\times10^{-3}$. Fig. \ref{thex2} shows that all the LOS rates, and the switching surfaces, $\sigma_i$, converge to zero at $t=20$ seconds, while Figs.~\ref{trex2} and \ref{samex2} show the trajectories of the five missiles, switching surface deviations and their lateral acceleration profiles, respectively. Observations similar to Scenario 1, including successful interception, are valid here as well. Note that the maximum value of lateral acceleration required by any agent is lower in this case than Scenario 1. This is because, here the gains corresponding to agents 3 and 4 have been reduced (these agents had the two highest gains in Scenario 1) and so the corresponding lateral acceleration values are also lower. Note that although the perturbed gains are not consistent with Algorithm \ref{algo2}, due to the insensitivity of the final consensus value and consensus time to variations in gain, there is no change in these values when two of the gains are perturbed. This further vindicates the claim in Remark \ref{robustness}. In general, the lateral acceleration requirements in both the scenarios can be further reduced if the desired consensus time, $t_f$, is increased by scaling down all the gains using the same factor.

\begin{figure*}[!ht]
	\centering\subfigure[Trajectories of missiles]{\includegraphics[scale=.6]{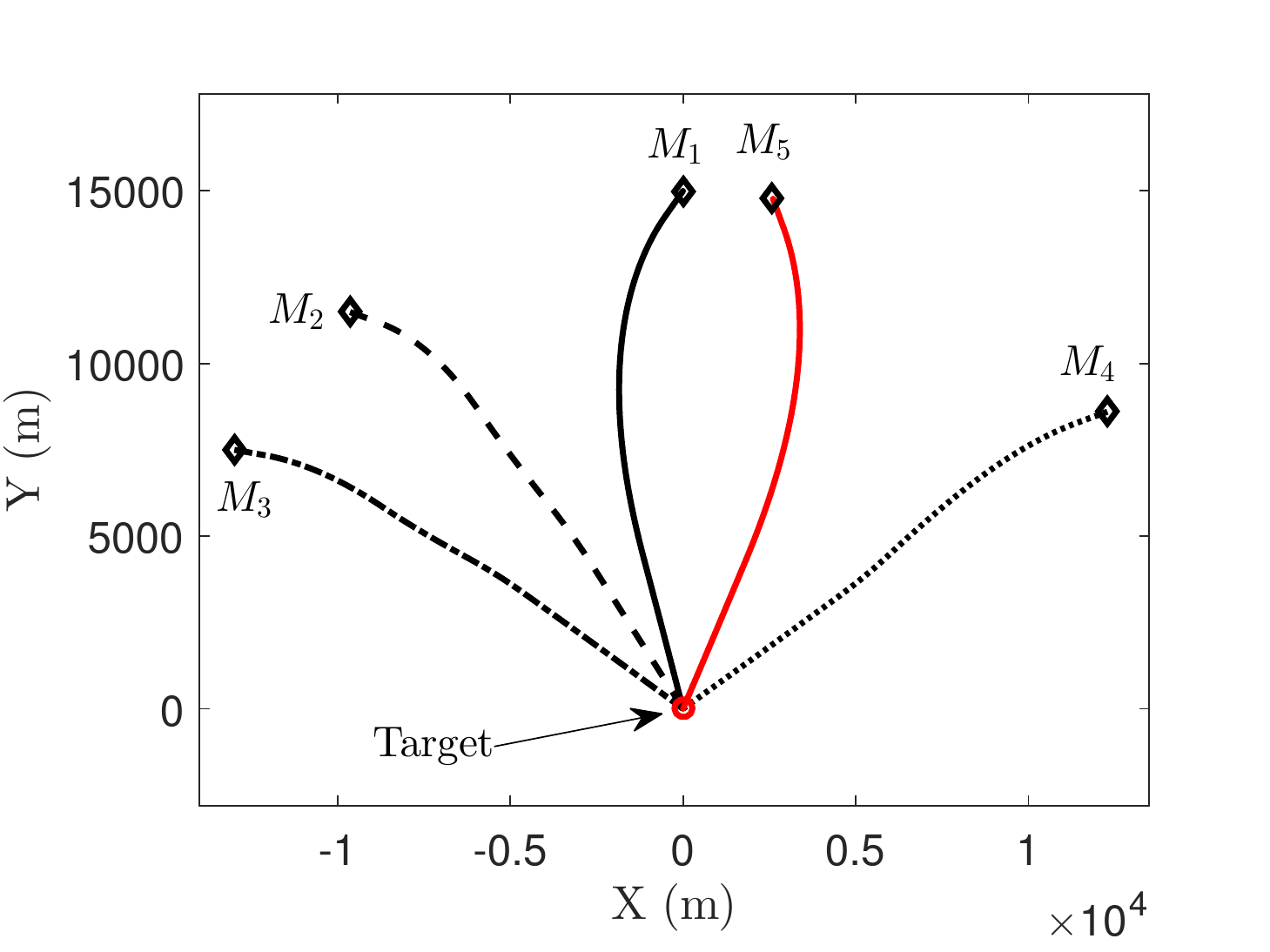}\label{trexreal1}}
	\subfigure[Switching surface deviations and lateral accelerations]{\includegraphics[scale=.45]{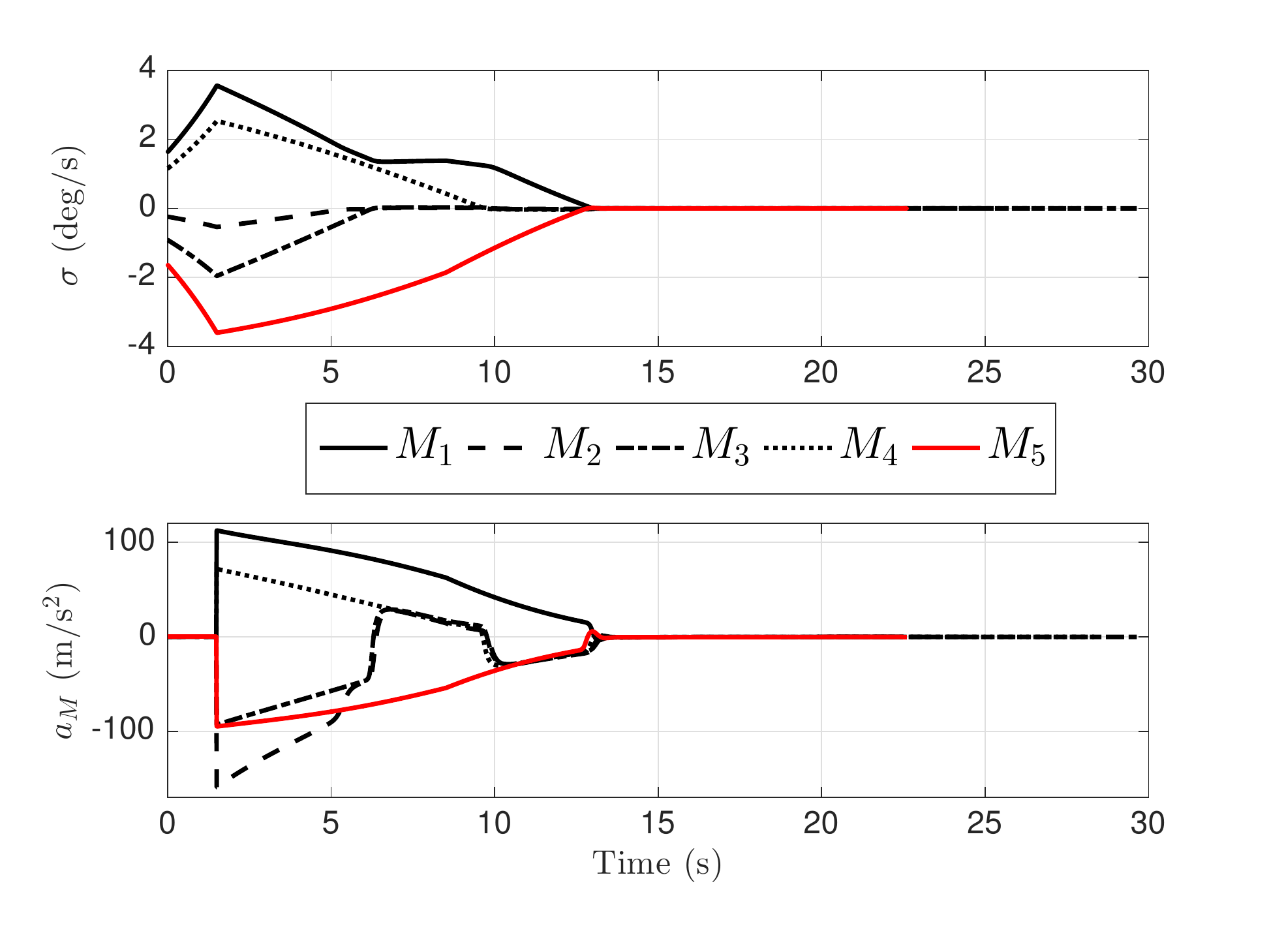}\label{samreal1}}
	\\
	\subfigure[LOS angles and their rates]{\includegraphics[scale=.6]{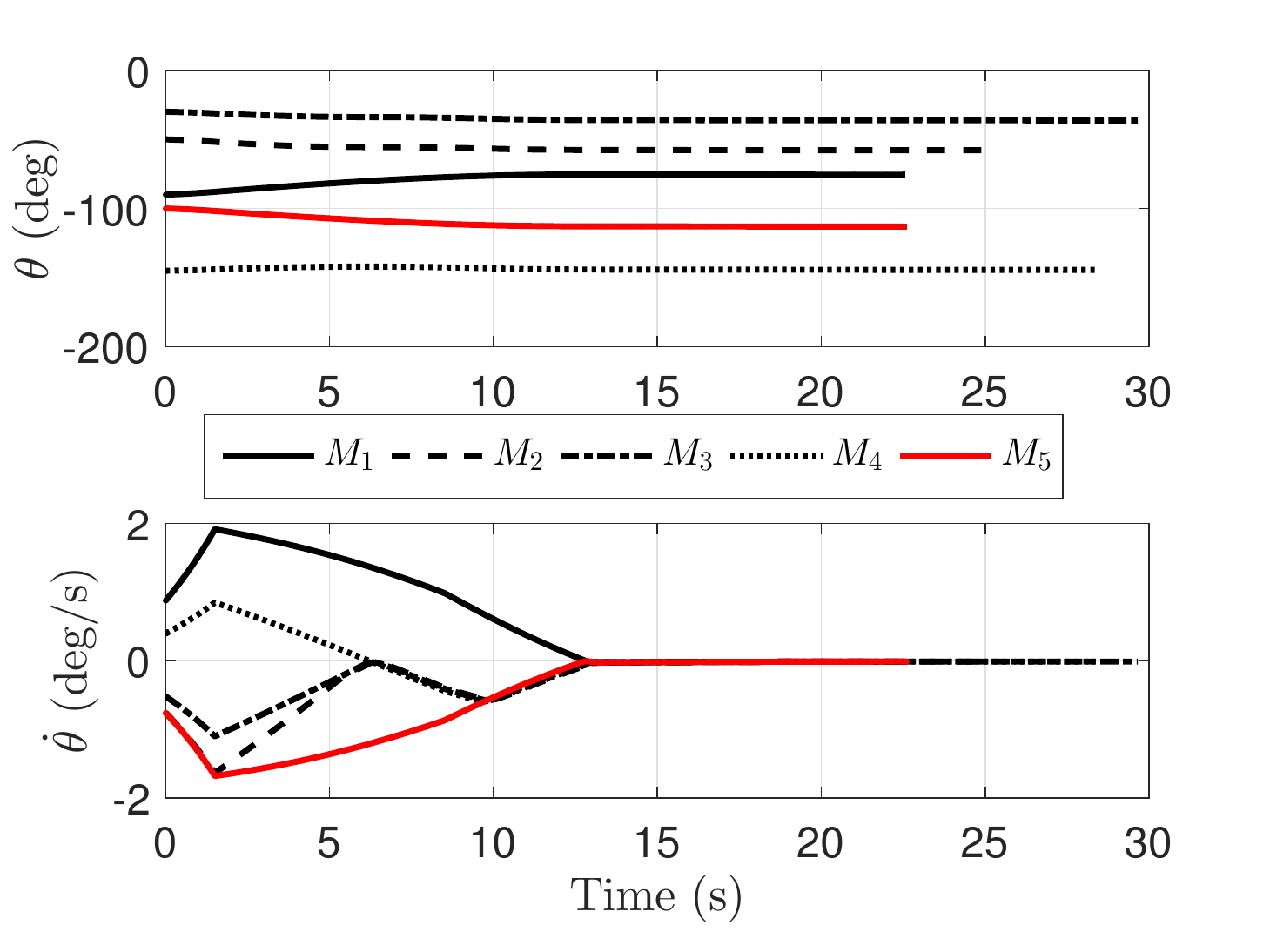}\label{threal1}}
	\subfigure[Missile speeds and drag profiles]{\includegraphics[scale=.6]{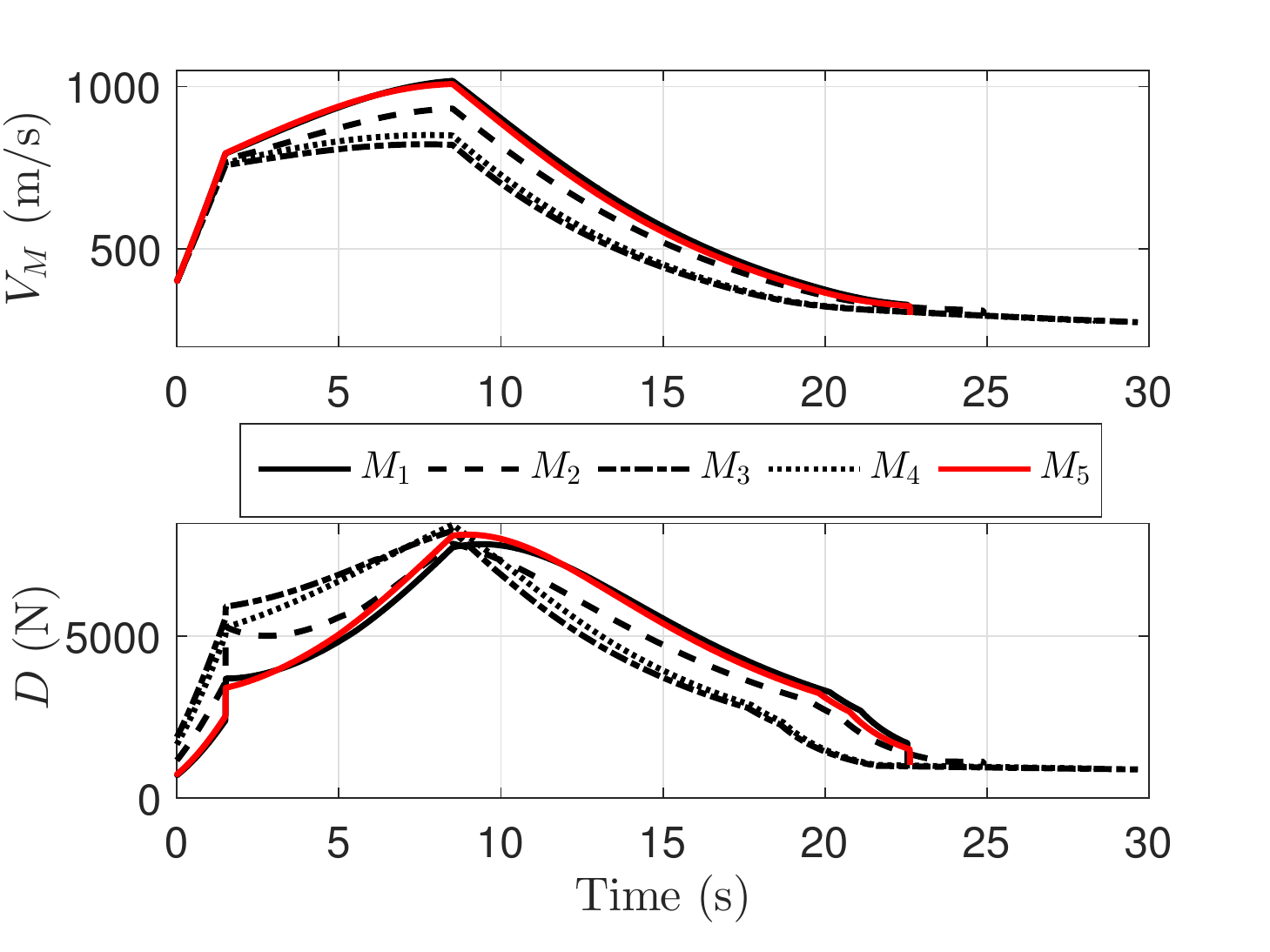}\label{vmdexreal1}}
	\caption{Cooperative guidance design with realistic missile models.}
	\label{figexreal}
\end{figure*}

\subsection*{Realistic Missile Models}

In this subsection, performance of the proposed guidance law is evaluated with a realistic interceptor model \cite{kee1998near}. It will be shown here that due to the inherent robustness of the SMC theory, the interceptor can achieve all objectives under a realistic scenario where the speed of the interceptor varies with time due to aerodynamic effects and time-limited thrust capability. Since the guidance law is designed for planar engagements, an interceptor model in the pitch plane is considered. The equations of motion of the $i^{\rm th}$ point-mass interceptor, flying over a flat and non-rotating Earth, are given by
\begin{equation}\label{vdotrl}
\hspace{-.1cm}\begin{aligned}
 {\dot X}_{M_i} =&~ {V_{M_i}}\cos \gamma _{M_i};~~ {\dot Y}_{M_i} = {V_{M_i}}\sin \gamma _{M_i},\\
 {\dot V}_{M_i} =&~ \frac{{T_i - D_i}}{m_i} - g\sin \gamma _{M_i};~~ {\dot \gamma }_{M_i} = \frac{a_{M_i} - g\cos  \gamma _{M_i} }{V_{M_i}},
\end{aligned}
\end{equation}
where $m_i, V_{M_i}$, and $\gamma_{M_i}$, are the mass, speed and flight path angle of the interceptor, respectively; $T_i$ and $D_i$ are thrust and drag on the interceptor, respectively; and $g$ is acceleration due to gravity. Also, $X_{M_i}$ and $Y_{M_i}$ denote the position of the interceptor in the Cartesian coordinate system and $a_{M_i}$ is the commanded interceptor lateral acceleration. The aerodynamic drag acting on the interceptor is modelled as
\begin{equation}\label{eq:drag}
\begin{aligned}
 D_i =&~ D_{0_i} + D_{I_i};~ {D_{0_i}} = C_{D_{0_i}}Q_i s_i,\\
 {D_{I_i}} =&~ \dfrac{{K_i{m_i^2}a_{M_i}^2}}{{Q_i s_i}},~K_i = \dfrac{1}{{\pi {A_{r_i}}e_i}};~~ Q_i = \dfrac{1}{2}\rho_i {V_{M_i}^2},
\end{aligned}
\end{equation}
where, $D_{0_i}$ and $D_{I_i}$, denote the zero-lift drag and induced drag, respectively; $C_{D_{0_i}}, K_i$, $A_{r_i}, e_i, \rho_i, s_i,$ and $Q_i$ are the zero-lift drag coefficient, induced drag coefficient, aspect ratio, efficiency factor, atmosphere density, reference area, and dynamic pressure, respectively. The detail of these coefficients and other parameters, and their variations can be found in \cite{kumar2014,kee1998near}.

In this simulation, the relative distance of all missile from the ground based target is 15 km. The LOS angles and initial flight path angles of missiles are given by $\vect{\theta}=[ -90^\circ~ -50^\circ~ -30^\circ~ -145^\circ~ -100^\circ ]$ and $\vect{\gamma}_M=[  -125^\circ~  -20^\circ~ -10 ^\circ~-160^\circ~ -70^\circ]$, respectively, which result in LOS rates, at the end of the boost phase, given by $[3.3208  ~ -2.8465  ~ -1.9160 ~   1.4691 ~  -2.9160]\times 10^{-2}$ rad/s. Note that the guidance loop is closed after the boost phase is over at $t=1.5$ seconds, and hence the values of LOS rates at 1.5 seconds serve as initial values for the design of gains. The gains are chosen using Algorithm~\ref{algo2} as $\vect{w} = [3.3208 ~6.8490~ 4.1900~ 3.3320~ 2.8778]\times 10^{-3}$ so as to achieve consensus after 10 seconds (11.5 seconds since start of the overall ebgagement). The initial speeds of all the missiles are equal to 400 m/s. Fig.~\ref{figexreal} shows the trajectories of all missiles, switching surface deviations and required lateral acceleration profiles, LOS angles and their rates, variations of missile's speeds, and the drag acting on them. The designed gains, which are supposed to lead to consensus at 11.5 seconds, result in consensus at around 12.5 seconds. This is on account of the variations in missile velocities, which were not considered during the design process of the gains. All the missiles are nevertheless able to achieve consensus in LOS rates at zero, as desired. This is shown in in Fig.~\ref{threal1}. Interception of the target occurs without mutual collision, while the acceleration demands are zero after consensus is achieved, as shown in Fig.~\ref{samreal1}. The missiles' speeds initially increase at a faster rate during the boost phase as shown in Fig.~\ref{vmdexreal1}, followed by relatively slower growth during sustain phase. During the terminal phase, the speeds decrease, due to presence of drag and lack of the thrust to compensate for the same, as shown in Fig.~\ref{vmdexreal1}. The drags are higher during sustain phase because of higher induced drag, owing to high missile speeds and lateral accelerations.


\section{Conclusions}\label{conc}

This paper presented a finite-time consensus algorithm for agents communicating over a cycle digraph. A detailed analysis was carried out for the case when the gains for the agents were allowed to be heterogeneous. Two algorithms were presented that aid in determining the point of convergence, where consensus occurs, and also the finite-time at which consensus occurs. 

This consensus scheme was then applied to a cooperative guidance scheme for multiple interceptors which ensures collision avoidance among the missiles. All the interceptors converged on their collision courses at the same time and remained there till interception. The applicability of the consensus scheme was tested on both realistic models of interceptors, as well as the constant speed models. Simulations established that the proposed strategy is effective in ensuring successful interception of the target while the missiles also avoid mutual collision. Future applications of this work may include ensuring specified relative spacing, and salvo attack for the engagement scenario, while theoretical studies of finite-time consensus laws maybe extended to consider higher order agent models over more general directed graphs.

\bibliographystyle{unsrt}
\bibliography{Bibtex_SRK_DM}

\end{document}